\documentclass[letterpaper,english,11pt]{article}
\usepackage[T1]{fontenc}
\usepackage[latin9]{inputenc}

\usepackage[top=1in, bottom=1in, left=1in, right=1in]{geometry}

\usepackage{verbatim}
\usepackage{float}
\usepackage{amssymb,amsfonts,amsmath}
\usepackage{enumerate}
\usepackage{color}
\usepackage{graphicx}
\usepackage{amsthm}
\usepackage{xspace}

\usepackage{caption}
\usepackage{subcaption}
\floatstyle{ruled}
\newfloat{Figure}{tbp}{loa}
\floatname{Figure}{Figure}

\newcommand{\cB}{{\cal B}}

\newcommand{\cT}{{\cal T}}
\newcommand{\cS}{{\cal S}}
\newcommand{\eps}{\varepsilon}
\newcommand{\tO}{{\tilde{O}}}

\newcommand{\rev}[1]{{#1}^{\textup{rev}}}

\newcommand{\List}{\mathcal{L}}

\newtheorem{lemma}{Lemma}

\newtheorem{theorem}{Theorem}
\newtheorem{conjecture}{Conjecture}
\newtheorem{definition}{Definition}

\DeclareMathOperator{\mmod}{mod}

\newcommand{\pipe}{|}
\newcommand{\plus}{+}
\newcommand{\types}{\ensuremath{\{\circ,\pipe,\star,\plus\}^*}}

\newcommand{\poly}{\mathrm{poly}}

\author{Karl Bringmann\thanks{Max Planck Institute for Informatics, Saarland Informatics Campus.
    Email: \texttt{kbringma@mpi-inf.mpg.de}
    }
\and
 Allan Gr\o nlund\thanks{Aarhus University.
   Email: \texttt{jallan@cs.au.dk}. Supported by Center for Massive Data Algorithmics, a Center of the Danish National Research Foundation, grant DNRF84.
 }
\and
  Kasper Green Larsen\thanks{Aarhus University.
    Email: \texttt{larsen@cs.au.dk}. Supported by Center for Massive Data Algorithmics, a Center of the Danish National Research Foundation, grant DNRF84, a Villum Young Investigator Grant and an AUFF Starting Grant.
  }
}
\date{}

\title{A Dichotomy for Regular Expression Membership Testing}

\begin{document}
\maketitle
\thispagestyle{empty}
\addtocounter{page}{-1}

\begin{abstract}
We study regular expression membership testing: Given a regular expression of size $m$ and a string of size $n$, decide whether the string is in the language described by the regular expression. Its classic $O(nm)$ algorithm is one of the big success stories of the 70s, which allowed pattern matching to develop into the standard tool that it is today. 

Many special cases of pattern matching have been studied that can be solved faster than in quadratic time. However, a systematic study of tractable cases was made possible only recently, with the first conditional lower bounds reported by Backurs and Indyk [FOCS'16]. Restricted to any ``type'' of homogeneous regular expressions of depth 2 or 3, they either presented a near-linear time algorithm or a quadratic conditional lower bound, with one exception known as the Word Break problem.

In this paper we complete their work as follows:
\begin{itemize}
  \item We present two almost-linear time algorithms that generalize all known almost-linear time algorithms for special cases of regular expression membership testing.
  \item We classify all types, except for the Word Break problem, into almost-linear time or quadratic time assuming the Strong Exponential Time Hypothesis. 
  This extends the classification from depth 2 and 3 to any constant depth.
  \item For the Word Break problem we give an improved $\tO(n m^{1/3} + m)$ algorithm. Surprisingly, we also prove a matching conditional lower bound for combinatorial algorithms. This establishes Word Break as the only intermediate problem.
\end{itemize}
In total, we prove matching upper and lower bounds for any type of bounded-depth homogeneous regular expressions, which yields a full dichotomy for regular expression membership testing.
\end{abstract}

\newpage

% !TEX root = main.tex

\section{Introduction}

A regular expression is a term involving an alphabet $\Sigma$ and the operations concatenation $\circ$, union~$\pipe$, Kleene's star $\star$, and Kleene's plus $\plus$, see Section~\ref{sec:preliminaries}. In regular expression \emph{membership testing}, we are given a regular expression $R$ and a string $s$ and want to decide whether $s$ is in the language described by $R$. In regular expression \emph{pattern matching}, we instead want to decide whether \emph{any substring} of $s$ is in the language described by $R$. 
A big success story of the 70s was to show that both problems have $O(nm)$ time algorithms~\cite{Thompson68}, where $n$ is the length of the string $s$ and $m$ is the size of $R$. This quite efficient running time, coupled with the great expressiveness of regular expressions, made pattern matching the standard tool that it is today. 

Despite the efficient running time of $O(nm)$, it would be desirable to have even faster algorithms. A large body of work in the pattern matching community was devoted to this goal, improving the running time by logarithmic factors \cite{Myers92,BT09} and even to near-linear for certain special cases \cite{KMP77,CH02,AC75}.

A systematic study of the complexity of various special cases of pattern matching and membership testing was made possible by the recent advances in the field of conditional lower bounds, where tight running time lower bounds are obtained via fine-grained reductions from certain core problems like satisfiability, all-pairs-shortest-paths, or 3SUM (see, e.g.,~\cite{AnkaOvermars,WAPSP, AW14, PatWil10}). Many of these conditional lower bounds are based on the Strong Exponential Time Hypothesis (SETH)~\cite{ImpagliazzoPaturi} which asserts that $k$-satisfiability has no $O(2^{(1-\eps)n})$ time algorithm for any $\eps>0$ and all $k \ge 3$. 

The first conditional lower bounds for pattern matching problems were presented by Backurs and Indyk~\cite{backursindyk}. Viewing a regular expression as a tree where the inner nodes are labeled by $\circ,\pipe,\star$, and $\plus$ and the leafs are labeled by alphabet symbols, they call a regular expression \emph{homogeneous of type $t \in \{\circ,\pipe,\star,\plus\}^d$} if in each level $i$ of the tree all inner nodes have type $t_i$, and the depth of the tree is at most $d$. Note that leafs may appear in any level, and the degrees are unbounded. This gives rise to natural restrictions \emph{$t$-pattern matching} and \emph{$t$-membership}, where we require the regular expression to be homogeneous of type $t$. The main result of Backurs and Indyk~\cite{backursindyk} is a characterization of $t$-pattern matching for all types $t$ of depth $d \le 3$: For each such problem they either design a near-linear time algorithm or show a quadratic lower bound based on SETH. We observed that the results by Backurs and Indyk actually even yield a classification for all $t$, not only for depth $d \le 3$. This is not explicitly stated in \cite{backursindyk}, so for completeness we prove it in this paper, see Appendix~\ref{app:patternmatching}. This closes the case for $t$-pattern matching.

For $t$-membership, Backurs and Indyk also prove a classification into near-linear time and ``SETH-hard'' for depth $d \le 3$, with the only exception being $\plus \pipe \circ$-membership. The latter problem is also known as the Word Break problem, since it can be rephrased as follows: Given a string $s$ and a dictionary $D$, can $s$ be split into words contained in $D$? Indeed, a regular expression of type $\circ$ represents a string, so a regular expression of type $\pipe \circ$ represents a dictionary, and type $\plus \pipe \circ$ then asks whether a given string can be split into dictionary words.
A relatively easy algorithm solves the Word Break problem in randomized time $\tO(n m^{1/2} + m)$, which Backurs and Indyk improved to randomized time $\tO(n m^{1/2 - 1/18} + m)$. 
Thus, the Word Break problem is the only studied special case of membership testing (or pattern matching) for which no near-linear time algorithm or quadratic time hardness is known. In particular, no other special case is ``intermediate'', i.e., in between near-linear and quadratic running time. 
Besides the status of Word Break, Backurs and Indyk also leave open a classification for $d>3$.

%Thus, the Word Break problem is in an interesting state, since it is not known to be solvable in near-linear time. Also, there is no other intermediate special case known for membership testing (or pattern matching), and recent improvements suggest that near-linear time might be possible. Moreover, Backurs and Indyk leave open a classification for $d>3$.
%
%In particular, Backurs and Indyk leave open whether the Word Break problem has a near-linear time algorithm. There 
%If not, then it takes superlinear time and thus is an intermediate problem in between linear and quadratic running time, which would be surprising. There is no other intermediate problem known for regular expression membership testing (or pattern matching), 
%
%However, this leaves the open problem of whether there is a near-linear time algorithm for the Word Break problem, which would finish the classification for $d \le 3$. Moreover, they leave open a classification for $d>3$.

\subsection{Our Results}

In this paper, we complete the dichotomy started by Backurs and Indyk~\cite{backursindyk} to a full dichotomy for any depth $d$. In particular, we (conditionally) establish Word Break as the only intermediate problem for (bounded-depth homogeneous) regular expression membership testing. More precisely, our results are as follows.

\paragraph{Word Break Problem.}
We carefully study the only depth-3 problem left unclassified by Backurs and Indyk. In particular, we improve Backurs and Indyk's  $\tO(n m^{1/2 - 1/18} + m)$ randomized algorithm to a deterministic $\tO(n m^{1/3} + m)$ algorithm.

\begin{theorem} \label{thm:wordbreakalgo}
  The Word Break problem can be solved in time $O(n (m \log m)^{1/3} + m)$. %$\tO(n m^{1/3} + m)$.
\end{theorem}
We remark that often running times of the form $\tilde O(n \sqrt{m})$ stem from a tradeoff of two approaches to a problem. Analogously, our time $\tilde O(n m^{1/3} + m)$ stems from trading off three approaches.

%Very surprisingly, we also prove a matching conditional lower bound, by presenting a reduction from $k$-Clique. 
%
%\begin{theorem} \label{thm:wordbreaklower}
%  For any $k \ge 4$, given a $k$-Clique instance on $n$ vertices, we can construct an equivalent Word Break instance on a string of length $O(n^{k-1})$ and a dictionary $D$ of total size $\|D\| = \sum_{d \in D} |d| = O(n^3)$. The reduction runs in linear time in the output size.
%\end{theorem}
%
%Note that via the above reduction an $O(n m^{1/3-\eps} + m)$ algorithm for the Word Break problem yields an $O(n^{k-3 \eps})$ algorithm for $k$-Clique. The fastest known 4-Clique algorithm uses 
%
%Recall that combinatorial algorithms are a notion without agreed upon definition, intuitively meaning that we forbid unpractical algorithms such as fast matrix multiplication.
%The $k$-Clique problem has a trivial $O(n^k)$ time algorithm, an $O(n^k / \lg^k n)$ combinatorial algorithm~\cite{vw09}, and all known faster algorithms use fast matrix multiplication~\cite{EG04}. 

Very surprisingly, we also prove a matching conditional lower bound. Our result only holds for combinatorial algorithms, which is a notion without agreed upon definition, intuitively meaning that we forbid unpractical algorithms such as fast matrix multiplication. We use the following hypothesis; a slightly weaker version has also been used in \cite{abboudLowerBoundForValiantsParser} for context free grammar parsing. Recall that the $k$-Clique problem has a trivial $O(n^k)$ time algorithm, an $O(n^k / \lg^k n)$ combinatorial algorithm~\cite{vw09}, and all known faster algorithms use fast matrix multiplication~\cite{EG04}. 

\begin{conjecture} \label{conj:clique}
  For all $k \ge 3$, any combinatorial algorithm for $k$-Clique takes time $n^{k-o(1)}$. 
\end{conjecture}

We provide a (combinatorial) reduction from $k$-Clique to the Word Break problem showing:

\begin{theorem} \label{thm:wordbreaklower}
  Assuming Conjecture~\ref{conj:clique}, the Word Break problem has no combinatorial algorithm in time $(n m^{1/3-\eps} + m)$ for any $\eps > 0$.
  %\karl{not proven: This holds even restricted to $m = \Theta(n^\alpha)$ for any constant $\alpha > 0$.}
\end{theorem}

This is a surprising result for multiple reasons. First, $n m^{1/3}$ is a very uncommon time complexity, specifically we are not aware of any other problem where the fastest known algorithm has this running time. Second, it shows that the Word Break problem is an \emph{intermediate} problem for $t$\nobreakdash-membership, as it is neither solvable in almost-linear time nor does it need quadratic time. Our results below show that the Word Break problem is, in fact, the \emph{only} intermediate problem for $t$-membership, which is quite fascinating. 

We leave it as an open problem to prove a matching lower bound without
the assumption of ``combinatorial''. Related to this question, note that the currently fastest algorithm for 4-Clique is based on
fast rectangular matrix multiplication and runs in time
$O(n^{3.256689})$~\cite{EG04,GL12}. If this bound is close to optimal, then we can still establish Word Break as an intermediate problem (without any restriction to combinatorial algorithms).

\begin{theorem} \label{thm:wordbreaklowertwo}
  For any $\delta > 0$, if 4-Clique has no $O(n^{3+\delta})$ algorithm, then Word Break has no
$O(n^{1+\delta/3})$ algorithm for $n=m$.
\end{theorem}

% for
%anyone not believing Conjecture~\ref{conj:clique} but willing to
%believe that the currently fastest algorithm for 4-Clique based on
%fast rectangular matrix multiplication that runs in
%$O(n^{3.256689})$~\cite{EG04,GL12} is close to optimal\footnote{no
%  closed formula in \cite{GL12}}, we still otain a polynomial lower
%bound. If 4-Clique takes time $n^{3+\delta}$ then word break takes
%$n^{1+\delta/3}$ for $n=m$, which still establishes Word Break as an
%intermediate problem.   \karl{consider making this a conjecture and a theorem}
%%
%%\begin{conjecture}\label{conj:4clique}
%% There is an \eps>0 such that any algorithm for 4-Clique takes $\Omega(3^{3+\eps})$ time
%% \end{conjecture}
%%
We remark that this situation of having matching conditional lower
bounds only for combinatorial algorithms is not uncommon, see, e.g.,
Sliding Window Hamming Distance~\cite{CIP09}.

\paragraph{New Almost-Linear Time Algorithms.}
We establish two more types for which the membership problem is in almost-linear time.

\begin{theorem} \label{thm:newalgos}
  We design a deterministic $\tO(n)+O(m)$ algorithm for $\pipe \plus \circ \plus$-membership and an expected time $n^{1+o(1)}+O(m)$ algorithm for $\pipe \plus \circ \pipe$-membership. 
These algorithms also work for $t$-membership for any subsequence $t$ of $\pipe \plus \circ \plus$ or $\pipe \plus \circ \pipe$, respectively. 
\end{theorem}

This generalizes \emph{all} previously known almost-linear time algorithms for any $t$-membership problem, as all such types $t$ are proper subsequences of $\pipe \plus \circ \plus$ or $\pipe \plus \circ \pipe$. Moreover, no further generalization of our algorithms is possible, as shown below.

\paragraph{Dichotomy.}
We enhance the classification of $t$-membership started by Backurs and Indyk for $d \le 3$ to a complete dichotomy for all types $t$. To this end, we first establish the following simplification rules.

\begin{lemma} \label{lem:simplify}
  For any type $t$, applying any of the following rules yields a type $t'$ such that $t$-membership and $t'$-membership are equivalent under linear-time reductions:
  \begin{enumerate}
    \item replace any substring $p p$, for any $p \in \{\circ,\pipe,\star,\plus\}$, by $p$,
    %\item replace any substring $\star \plus$ by $\star$,
    %\item replace any substring $\star \pipe \plus$ (or $\star \pipe \star$) by $\star \pipe$,
    \item replace any substring $\plus \pipe \plus$ by $\plus \pipe$,
    \item replace prefix $r \star$ by $r \plus$ for any $r \in \{\plus,\pipe\}^*$.
  \end{enumerate}
  We say that $t$-membership \emph{simplifies} if one of these rules applies.
  Applying these rules in any order will eventually lead to an unsimplifiable type.
\end{lemma}

We show the following dichotomy. Note that we do not have to consider simplifying types, as they are equivalent to some unsimplifiable type.

\begin{theorem} \label{thm:main}
  For any $t \in \types$ one of the following holds:
  \begin{itemize}
    \item $t$-membership simplifies,
    \item $t$ is a subsequence of $\pipe \plus \circ \plus$ or $\pipe \plus \circ \pipe$, and thus $t$-membership is in almost-linear time (by Theorem~\ref{thm:newalgos}),
    \item $t = \plus \pipe \circ$, and thus $t$-membership is the Word Break problem taking time $(n m^{1/3} + m)^{1 \pm o(1)}$ (by Theorems~\ref{thm:wordbreakalgo} and~\ref{thm:wordbreaklower}, assuming Conjecture~\ref{conj:clique}), or
    \item $t$-membership takes time $(nm)^{1-o(1)}$, assuming SETH.
  \end{itemize}
\end{theorem}

This yields a complete dichotomy for any constant depth $d$. We discussed the algorithmic results and the results for Word Break before. Regarding the hardness results,
Backurs and Indyk~\cite{backursindyk} gave SETH-hardness proofs for $t$-membership on types $\circ \star$, $\circ \pipe \circ$, $\circ \plus \circ$, $\circ \pipe \plus$, and $\circ \plus \pipe$. We provide further SETH-hardness for types $\plus \pipe \circ \plus$, $\plus \pipe \circ \pipe$, and $\pipe \plus \pipe \circ$. 
To get from these (hard) \emph{core types} to all remaining hard types, we would like to argue that all hard types contain one of the core types as a \emph{subsequence} and thus are at least as hard. However, arguing about subsequences fails in general, since the definition of ``homogeneous with type $t$'' does not allow to leave out layers. This makes it necessary to proceed in a more ad-hoc way.
%formally this intuition fails: It is not true that for any subsequence $t'$ of $t$ there is an easy linear-time reduction from $t'$-membership to $t$-membership. The reason is that the definition of ``homogeneous with type $t$'' does not allow edges in the regular expression tree to jump over a layer, and thus we may not drop any element of the sequence $t$, to embed a regular expression corresponding to a subsequence $t'$ of $t$.
%Nevertheless, we show that for any remaining SETH-hard type $t$ we can reduce one of the above core types to $t$-membership (via linear-time reductions). To this end we use a large case distinction, establishing that we need a reduction from $t'$-membership to $t$-membership only for very restricted forms of subsequences $t'$ of $t$. 

In summary, we provide matching upper and lower bounds for any type of bounded-depth homogeneous regular expressions, which yields a full dichotomy for the membership problem.

%\subsection{Technical Contribution}
%\karl{What does this need to be. A summary of our techniques. That seems hard}
%...
%
%time complexities like $\tO(n m^{1/3})$ are uncommon. Often, a time complexity like $\tO(n m^{1/2})$ comes up from either of two situations: (1) we obtain an algorithm with running time $\tO(n p + n m / p)$ for a parameter $p$ that we may set and thus optimize to $\sqrt{m}$, or (2) we trade off two approaches with running times, say, $\tO(n q)$ and $\tO(n m / q)$ for some parameter $q$ depending on the input. Our running time of $\tO(n m^{1/3})$ stems from a three-wise situation: For parameters $p$ (that we may set) and $q$ (depending on the input) we obtain algorithms with running times $\tO(n q)$ and $\tO(n p + nm/(pq))$. The same kind of optimization as in the usual situation yields a running time of $\tO(n m^{1/3})$, not including a preprocessing time of $O(n+m)$.
%
%...

\subsection{Organization}
The paper is organized as follows. We start with preliminaries in Section~\ref{sec:preliminaries}.
For the Word Break problem, we prove the conditional lower bounds in Section \ref{sec:wb_lower}, followed by the matching upper
bound in Section \ref{sec:wb_upper}. In Section \ref{sec:almostlinear}
we present our new almost-linear time algorithms for
two types, and in Section \ref{sec:sethlower} we prove SETH-hardness for three types.
Finally, in Section \ref{sec:dichotomy} we
prove that our results yield a full dichotomy for homogenous regular
expression membership testing.

% !TEX root = main.tex

\section{Preliminaries}
\label{sec:preliminaries}

A regular expression is a tree with leafs labelled by symbols in an alphabet $\Sigma$ and inner nodes labelled by $\circ$ (at least one child), $\pipe$ (at least one child), $\plus$ (exactly one child), or $\star$ (exactly one child)\footnote{\label{foot:shortdiscussion}All our algorithms work in the general case where $\circ$ and $\pipe$ may have degree 1. For the conditional lower bounds, it may be unnatural to allow degree 1 for these operations. If we restrict to degrees at least 2, it is possible to adapt our proofs to prove the same results, but this is tedious and we think that the required changes would be obscuring the overall point. We discuss this issue in more detail when it comes up, see footnote~\ref{foot:discussion} on page~\pageref{foot:discussion}.}. 
The language described by the regular expression is recursively defined as follows. A leaf $v$ labelled by $c \in \Sigma$ describes the language $L(v) := \{c\}$, consisting of one word of length~1. Consider an inner node $v$ with children $v_1,\ldots,v_\ell$. If $v$ is labelled by $\circ$ then it describes the language $\{ s_1 \ldots s_\ell \mid s_1 \in L(v_1), \ldots, s_k \in L(v_\ell)\}$, i.e., all concatenations of strings in the children's languages. If $v$ is labelled by $\pipe$ then it describes the language $L(v) := L(v_1) \cup \ldots L(v_\ell)$. If $v$ is labelled $\plus$, then its degree $\ell$ must be 1 and it describes the language $L(v) := \{ s_1 \ldots s_k \mid k \ge 1 \text{ and } s_1,\ldots, s_k \in L(v_1)\}$, and if $v$ is labelled $\star$ then the same statement holds with ``$k \ge 1$'' replaced by ``$k \ge 0$''. we say that a string $s$ \emph{matches} a regular expression $R$ if $s$ is in the language described by $R$.

We use the following definition from~\cite{backursindyk}. 
We let $\types$ be the set of all finite sequences over $\{\circ,\pipe,\star,\plus\}$; we also call this the set of \emph{types}.
For any $t \in \types$ we denote its length by $|t|$ and its $i$-th entry by $t_i$. We say that a regular expression is \emph{homogeneous of type $t$} if it has depth at most $|t|+1$ (i.e., any inner node has level in $\{1,\ldots,|t|\}$), and for any $i$, any inner node in level $i$ is labelled by $t_i$. We also say that the type of any inner node at level $i$ is $t_i$. This does not restrict the appearance of leafs in any level. 
 
\begin{definition}
  \label{def:lineartimereduction}
A linear-time reduction from $t$-membership to $t'$-membership is an
algorithm that, given a regular expression $R$ of type $t$ and
length $m$ and a string $s$ of length $n$, in total time $O(n+m)$
outputs a regular expression $R'$ of type $t'$ and size $O(m)$,
and a string $s'$ of length $O(n)$ such that $s$ matches $R$ if
and only if $s'$ matches $R'$.
\end{definition}

The Strong Exponential Time Hypothesis (SETH) was introduced by
 Impagliazzo, Paturi, and Zane~\cite{ImpagliazzoPZ01} and is defined as follows.
\begin{conjecture}
\label{conj:seth}
  For no $\varepsilon>0$, k-SAT can be solved in time $O(2^{(1-\varepsilon)N})$
  for all $k\geq 3$.
\end{conjecture}
 Very often it is easier to show SETH-hardness based on the intermediate problem Orthogonal Vectors
 (OV): Given two sets of $d$-dimensinal vectors $A,B \subseteq
 \{0,1\}^d$ with $|A|=|B|=n$, determine if there exist
 vectors $a\in A, b\in B$ such that $\sum_{i=1}^d a[i] \cdot b[i] =
 0$.
 The following OV-conjecture follows from SETH \cite{RW05}.
 %\cite{RW05,BK15}.
 %
 \begin{conjecture}
   \label{conj:ov}
   For any $\varepsilon >0$ there is no
   algorithm for OV that runs in time
   $O(n^{2-\varepsilon} \poly(d))$.
 \end{conjecture}
 
% \begin{conjecture}
%   Let $0 < \alpha \leq 1$, $m=O(n^\alpha)$ and $d\leq n^{o(1)}$. For
%   any $\varepsilon>0$ there is no algorithm for Orthogonal Vectors
%   that runs in time $O((nm)^{1-\varepsilon})$.
% \end{conjecture}

% Cite from Backurs and Indyk:
To start off the proof for the dichotomy, we have the following hardness
results from~\cite{backursindyk}.
\begin{theorem} \label{thm:backursindyklower}
  For any type $t$ among $\circ \star$, $\circ \pipe \circ$, $\circ \plus \circ$, $\circ \pipe \plus$, and $\circ \plus \pipe$, any algorithm for $t$-membership takes time $(nm)^{1-o(1)}$ unless SETH fails.
\end{theorem}

% !TEX root = main.tex

\section{Conditional Lower Bound for Word Break}
\label{sec:wb_lower}

In this section we prove our conditional lower bounds for the Word Break problem, Theorems~\ref{thm:wordbreaklower} and~\ref{thm:wordbreaklowertwo}. Both theorems follow from the following reduction.

\begin{theorem} \label{thm:wordbreakreduction}
  For any $k \ge 4$, given a $k$-Clique instance on $n$ vertices, we can  construct an equivalent Word Break instance on a string of length $O(n^{k-1})$ and a dictionary $D$ of total size $\|D\| = \sum_{d \in D} |d| = O(n^3)$. The reduction is combinatorial and runs in linear time in the output size.
\end{theorem}

First let us see how this reduction implies Theorems~\ref{thm:wordbreaklower} and~\ref{thm:wordbreaklowertwo}.

\begin{proof}[Proof of Theorem~\ref{thm:wordbreaklower}]
  Suppose for the sake of contradiction that Word Break can be solved combinatorially in time $O(n m^{1/3-\eps} + m)$. Then our reduction yields a combinatorial algorithm for $k$-Clique in time $O( n^{k-1} \cdot (n^3)^{1/3 - \eps}) = O(n^{k-3 \eps})$, contradicting Conjecture~\ref{conj:clique}. This proves the theorem.
\end{proof}

\begin{proof}[Proof of Theorem~\ref{thm:wordbreaklower}]
  Assuming that 4-Clique has no $O(n^{3+\delta})$ algorithm for some $\delta > 0$, we want to show that Word Break has no $O(n^{1+\delta/3})$ algorithm for $n=m$. 
  
  Setting $k = 4$ in the above reduction yields a string and a dictionary, both of size $O(n^3)$ (which can be padded to the same size). Thus, an $O(n^{1+\delta/3})$ algorithm for Word Break with $n=m$ would yield an $O(n^{3+\delta})$ algorithm for 4-Clique, contradicting the assumption.
\end{proof}

It remains to prove Theorem~\ref{thm:wordbreakreduction}.
Let $G = (V,E)$ be an $n$-node graph on which we want to determine whether
there is a $k$-clique. The main idea of our reduction is to construct a
gadget that for any $(k-2)$-clique $S \subset V$ can determine whether
there are two nodes $u,v \in V \setminus S$ such that $(u,v) \in E$
and both $u$ and $v$ are connected to all nodes in $S$, i.e., $S \cup
\{u,v\}$ forms a $k$-clique in $G$. For intuition, we first present a
simplified version of our gadgets and then show how to modify them to
obtain the final reduction.

\paragraph{Simplified Neighborhood Gadget.}
Given a $(k-2)$-clique $S$, the purpose of our first gadget is to test
whether there is a node $u \in
V$ that is connected to all nodes in $S$. Assume the nodes in
$V$ are denoted
$v_1,\dots,v_n$. The alphabet $\Sigma$ over which we construct strings
has a symbol $i$ for each $v_i$. Furthermore, we assume $\Sigma$ has
special symbols $\#$ and $\$$. The simplified neighborhood gadget for $S = \{v_{i_1},\dots,v_{i_{k-2}}\}$
has the text $T$ being 
$$
\$123\cdots n\#i_1\#123\cdots n \#i_2\#123 \cdots n\# \cdots n \# i_{k-2}\# 123\cdots n\$
$$
and the dictionary $D$ contains for every edge $(v_i, v_j) \in E$, the
string:
$$
i(i+1)\cdots n \# j \# 123\cdots (i-2)(i-1)
$$
and for every node $v_i$, the two strings
$$
\$123 \cdots
(i-2)(i-1)
$$ 
and
$$
i(i+1)\cdots n\$
$$
The idea of the above construction is as follows: Assume we want to
break $T$ into words.  The crucial observation is that to match $T$ using $D$, we have to
start with $\$123 \cdots (i-2)(i-1)$ for some node $v_i$. The only way we can possibly
match the following part $i(i+1)\cdots n \#i_1\#$ is if $D$ has the
string $i(i+1)\cdots n \# i_1 \# 123 \cdots (i-2)(i-1)$. But this is the
case if and only if $(v_i,v_{i_1}) \in E$, i.e. $v_{i}$ is a neighbor
of $v_{i_1}$. If indeed this is the case, we have now matched the
prefix $\# i_1 \#1 2 3\cdots (i-2)(i-1)$ of the next block. This means
that we can still only use strings starting with $i(i+1)\cdots$ from
$D$. Repeating this argument for all $v_{i_j} \in S$, we conclude that we
can break $T$ into words from $D$ if and only if there is some node $v_i$ that is a neighbor of every node $v_{i_j} \in S$.

\paragraph{Simplified $k$-Clique Gadget.}
With our neighborhood gadget in mind, we now describe the main ideas
of our gadget that for a given $(k-2)$-clique $S$ can test whether
there are two nodes $v_i,v_j$ such that $(v_i,v_j) \in E$ and $v_i$
and $v_j$ are both connected to all nodes of $S$, i.e., $S \cup \{v_i,
v_j\}$ forms a $k$-clique. 

Let $T_S$ denote the text used in the neighborhood gadget for $S$,
i.e.
$$
T_S = \$123\cdots n\#i_1\#123\cdots n \#i_2\#123 \cdots n\# \cdots n \# i_{k-2}\# 123\cdots n\$
$$
Our $k$-clique gadget for $S$ has the following text $T$:
$$
T_S\gamma T_S
$$
where $\gamma$ is a special symbol in $\Sigma$. The dictionary $D$ has
the strings mentioned in the neighborhood gadget, as well as the
string
$$
i(i+1)\cdots n\$ \gamma \$123 \cdots (j-1)
$$
for every edge $(v_i,v_j) \in E$. The idea of this gadget is as
follows: Assume we want to break $T$ into words. We have to start
using the dictionary string $\$123\cdots (i-1)$ for some node $v_i$. For such a
candidate node $v_i$, we can match the prefix
$$
\$123\cdots n\#i_1\#123\cdots n \#i_2\#123 \cdots n\# \cdots n \# i_{k-2}\#
$$
of $T_S \gamma T_S$ if and only if $v_i$ is a neighbor of every node
in $S$. Furthermore, the only way to match this prefix (if we start
with $\$123\cdots (i-1)$) covers
precisely the part:
$$
\$123\cdots n\#i_1\#123\cdots n \#i_2\#123 \cdots n\# \cdots n \#
i_{k-2}\#123 \cdots (i-2)(i-1)
$$
Thus if we want to also match the $\gamma$, we can only use strings
$$
i(i+1)\cdots n \$ \gamma \$123 \cdots (j-1)
$$
for an edge $(v_i, v_j) \in E$. Finally, by the second neighborhood
gadget, we can match the whole string $T_S \gamma T_S$ if and only if
there are some nodes $v_i,v_j$ such that $v_i$ is a neighbor of every
node in $S$ (we can match the first $T_S$), and $(v_i,v_j) \in E$ (we
can match the $\gamma$) and $v_j$ is a neighbor of every
node in $S$ (we can match the second $T_S$), i.e., $S \cup \{v_i, v_j
\}$ forms a $k$-clique.

\paragraph{Combining it all.} The above gadget allows us
to test for a given $(k-2)$-clique $S$ whether there are some two nodes
$v_i$ and $v_j$ we can add to $S$ to get a $k$-clique. Thus, our next
step is to find a way to combine such gadgets for all the
$(k-2)$-cliques in the input graph. The challenge is to compute an
\textit{OR} over all of them, i.e. testing whether at least one can be
extended to a $k$-clique. For this, our idea is to replace every
symbol in the above constructions with 3 symbols and then carefully
concatenate the gadgets. When we start matching the string $T$ against
the dictionary, we are matching against the first symbol of the first
$(k-2)$-clique gadget, i.e. we start at an offset of zero. We want to add strings to the dictionary that
\emph{always} allow us to match a clique gadget if we have an offset
of zero. These strings will then leave us at offset zero in the next
gadget. Next, we will add a string that allow us to change from offset
zero to offset one. We will then ensure that if we have an offset of
one when starting to match a $(k-2)$-clique gadget, we can only match
it if that
clique can be extended to a $k$-clique. If so, we ensure that we will
start at an offset of two in the next gadget. Next, we will also
add strings to the dictionary that allow us to match any gadget if we
start at an offset of two, and these strings will ensure we continue
to have an offset of two. Finally, we append symbols at the end of the
text that can \emph{only} be matched if we have an offset of two after
matching the last gadget. To summarize: Any breaking of $T$ into words
will start by using an offset of zero and simply skipping over
$(k-2)$-cliques that cannot be extended to a $k$-clique. Then once a
proper $(k-2)$-clique is found, a string of the dictionary is used to
change the start offset from zero to one. Finally, the clique is
matched and leaves us at an offset of two, after which the remaining
string is matched while maintaining the offset of two.

We now give the final details of the reduction. Let $G=(V,E)$ be the
$n$-node input graph to $k$-clique. We do as follows:
\begin{enumerate}
\item Start by iterating over every set of $(k-2)$ nodes $S$ in
  $G$. For each such set of nodes, test whether they form a
  $(k-2)$-clique in $O(k^2)$ time. Add each found $(k-2)$-clique $S$
  to a list $\List$. 
\item Let $\alpha, \beta, \gamma, \mu, \#$ and $\$$ be special symbols in the
  alphabet. For a string $T=t_1t_2t_3\cdots t_m$, let
$$
[T]^{(0)}_{\alpha, \beta} = \alpha t_1 \beta \alpha t_2 \beta \cdots
\alpha t_m \beta
$$
and
  $$
[T]^{(1)}_{\alpha,\beta} = t_1 \beta \alpha  t_2 \beta \alpha t_3 \beta \cdots
  \alpha t_m \beta \alpha
$$
For each node $v_i \in V$, add the following two strings
  to the dictionary $D$:
$$
[\$ 1 2 3 \cdots
(i-2) (i-1)]^{(1)}_{\alpha,\beta}
$$ 
and
$$
[i(i+1)\cdots n]^{(1)}_{\alpha,\beta}
$$
\item For each edge $(v_i,v_j) \in E$, add the following two strings
  to the dictionary:
$$
[i(i+1)\cdots n \# j \# 123\cdots (i-2)(i-1)]^{(1)}_{\alpha, \beta}
$$
and
$$
[i(i+1)\cdots n\$ \gamma \$123 \cdots (j-1)]^{(1)}_{\alpha, \beta}
$$
\item For each symbol $\sigma$ amongst $\{1,\dots, n, \$, \#, \gamma, \mu\}$, add the
  following two string to $D$:
$$
\alpha \sigma \beta
$$
and
$$
\beta \alpha \sigma
$$
Intuitively, the first of these strings is used for skipping a gadget
if we have an offset of zero, and the second is used for skipping a gadget
if we have an offset of two.
\item Also add the three strings 
$$
\alpha \mu \beta \alpha\ \  \$ \beta  \alpha \mu\ \  \beta \mu \mu
$$
to the dictionary. The first is intuitively used for changing from an
offset of zero to an offset of one (begin matching a clique gadget),
the second is used for changing from an offset of one to an offset of
two in case a clique gadget could be matched, and the last string is
used for matching the end of $T$ if an offset of two has been achieved.
\item We are finally ready to describe the text $T$. For a
  $(k-2)$-clique $S=\{v_{i_1},\dots,v_{i_{k-2}}\}$, let $T_S$ be the neighborhood gadget from above,
  i.e.
$$
T_S := \$123\cdots n\#i_1\#123\cdots n \#i_2\#123 \cdots n\# \cdots n \# i_{k-2}\# 123\cdots n\$
$$
For each $S \in \List$ (in some arbitrary order), we append the
string:
$$
[\mu T_S \gamma T_S \mu]^{(0)}_{\alpha, \beta}
$$
to the text $T$. Finally, once all these strings have been appended,
append another two $\mu$'s to $T$. That is, the text $T$ is:
$$
T:= \left(\circ_{S \in \List} [\mu T_S \gamma T_S \mu]^{(0)}_{\alpha,
    \beta}\right) \mu \mu
$$
\end{enumerate}
We want to show that the text $T$ can be broken into words from the
dictionary $D$ iff there is a $k$-clique in the input graph. Assume
first there is a $k$-clique $S$ in $G=(V,E)$. Let $S'$ be an arbitrary
subset of $k-2$ nodes from $S$. Since these form a $(k-2)$-clique, it
follows that $T$ has the substring $[\mu T_{S'} \gamma T_{S'}
\mu]^{(0)}_{\alpha, \beta}$. To
match $T$ using $D$, do as follows: For each $S''$ preceeding $S'$ in
$\List$, keep using the strings $\alpha \sigma \beta$ from step
4 above to match. This allows us to match everything preceeding $[\mu T_{S'} \gamma T_{S'}
\mu]^{(0)}_{\alpha, \beta}$ in $T$. Then use the string $\alpha \mu
\beta \alpha$ to match the beginning of $[\mu T_{S'} \gamma T_{S'}
\mu]^{(0)}_{\alpha, \beta}$. Now let $v_i$ and $v_j$ be the two nodes
in $S \setminus S'$. Use the string $[\$123 \cdots
(i-2)(i-1)]^{(1)}_{\alpha, \beta}$ to match the next part of $[\mu T_{S'} \gamma T_{S'}
\mu]^{(0)}_{\alpha, \beta}$. Then since $S$ is a $k$-clique, we have
the string 
$[i(i+1)\cdots n \# h \# 123\cdots (i-2)(i-1)]^{(1)}_{\alpha, \beta}$
in the dictionary for every $v_h \in S'$. Use these strings for
each $v_h \in S'$. Again, since $S$ is a $k$-clique, we also have the
edge $(v_i, v_j) \in E$. Thus we can use the string $$
[i(i+1)\cdots n\$ \gamma \$123 \cdots (j-1)]^{(1)}_{\alpha, \beta}
$$ to match across the $\gamma$ in $[\mu T_{S'} \gamma T_{S'}
\mu]^{(0)}_{\alpha, \beta}$. We then repeat the argument for $v_j$ and
repeatedly use the strings 
$[j(j+1)\cdots n \# h \# 123\cdots (j-2)(j-1)]^{(1)}_{\alpha, \beta}$
to match the second $T_{S'}$. We finish by using the string $[j(j+1) \cdots n]_{\alpha, \beta}^{(1)}$
followed by using $\$ \beta \alpha \mu$. We are now at an offset where we can
repeatedly use $\beta \alpha \sigma$ to match across all remaining $[\mu T_{S''} \gamma T_{S''}
\mu]^{(0)}_{\alpha, \beta}$. Finally, we can finish the match by using
$\beta \mu \mu$ after the last substring $[\mu T_{S''} \gamma T_{S''}
\mu]^{(0)}_{\alpha, \beta}$.

For the other direction, assume it is possible to break $T$ into words
from $D$. By construction, the last word used has to be $\beta \mu
\mu$. Now follow the matching backwards until a string not of the form
$\beta \alpha \sigma$ was used. This must happend eventually since $T$
starts with $\alpha$. We are now at a position in $T$ where the
suffix can be matched by repeatedly using $\beta \alpha \sigma$, and
then ending with $\beta \mu \mu$. By construction, $T$ has $\alpha
\sigma$ just before this suffix for some $\sigma  \in \{1,\dots, n,
\$, \#, \gamma, \mu\}$. The only string in $D$ that
could match this without being of the form $\beta \alpha \sigma$ is
the one string $\$\beta \alpha \mu$. It follows that we must be at the
end of some substring $[\mu T_{S'} \gamma T_{S'}
\mu]^{(0)}_{\alpha, \beta}$ and used $\$\beta \alpha \mu$ for matching
the last $\mu$. To match the preceeding $n$ in the last $T_{S'}$, we
must have used a string $[j(j+1)\cdots n]_{\alpha, \beta}^{(1)}$ for
some $v_j$. The only strings that can be used preceeding this are
strings of the form
$[j(j+1)\cdots n \# h \# 123\cdots (j-2)(j-1)]^{(1)}_{\alpha,
  \beta}$. Since we have matched $T$, it follows that $(v_j, v_h)$ is in
$E$ for every $v_h \in S'$. Having traced back the match across the
last $T_{S'}$ in $[\mu T_{S'} \gamma T_{S'}
\mu]^{(0)}_{\alpha, \beta}$, let $v_i$ be the node such that the
string $[i(i+1)\cdots n\$ \gamma \$123 \cdots (j-1)]^{(1)}_{\alpha,
  \beta}$ was used to match the $\gamma$. It follows that we must have
$(v_i,v_j) \in E$. Tracing the matching through the first $T_{S'}$ in $[\mu T_{S'} \gamma T_{S'}
\mu]^{(0)}_{\alpha, \beta}$, we conclude that we must also have $(v_i,
v_h) \in E$ for every $v_h \in S'$. This establishes that
$S' \cup \{v_i,v_j\}$ forms a $k$-clique in $G$.

\paragraph{Finishing the proof.} From the input graph $G$, we constructed the
Word Break instance in time $O(n^{k-2}k^2)$ plus the time needed to
output the text and the dictionary.  For every edge $(v_i,v_j) \in E$,
we added two strings to $D$, both of length $O(n)$. Furthermore, $D$
had two $O(n)$ length strings for each node $v_i \in V$ and another
$O(n)$ strings of constant length. Thus the total length of the
strings in $D$ is $M = O(|E|n+n) = O(n^3)$. The text $T$ has the
substring $[\mu T_{S'} \gamma T_{S'} \mu]^{(0)}_{\alpha, \beta}$ for
every $(k-2)$-clique $S$.  Thus $T$ has length $N = O(n^{k-1})$
(assuming $k$ is constant).
The entire reduction takes $O(n^{k-1}+n^3)$ time for constant $k$.
This finishes the reduction and proves Theorem~\ref{thm:wordbreakreduction}.

% !TEX root = main.tex

\section{Algorithm for Word Break}
\label{sec:wb_upper}

In this section we present an $\tO(n m^{1/3} + m)$ algorithm for the Word Break problem, proving Theorem~\ref{thm:wordbreakalgo}. Our algorithm uses many ideas of the randomized $\tO(n m^{1/2 - 1/18} + m)$ algorithm by Backurs and Indyk~\cite{backursindyk}, in fact, it can be seen as a cleaner execution of their main ideas. 
Recall that in the Word Break Problem we are given a set of strings $D = \{d_1,\ldots,d_k\}$ (the dictionary) and a string $s$ (the text) and we want to decide whether $s$ can be ($D$-)\emph{partitioned}, i.e., whether we can write $s = s_1 \ldots s_r$ such that $s_i \in D$ for all $i$. We denote the length of $s$ by $n$ and the total size of $D$ by $m := \|D\| := \sum_{i=1}^k |d_i|$.

We say that we can ($D$-)\emph{jump} from $j$ to $i$ if the substring $s[j+1..i]$ is in $D$. Note that if $s[1..j]$ can be partitioned and we can jump from $j$ to $i$ then also $s[1..i]$ can be partitioned. Moreover, $s[1..i]$ can be partitioned if and only if there exists $0 \le j < i$ such that $s[1..j]$ can be partitioned and we can jump from $j$ to $i$. For any power of two $q \ge 1$, we let $D_q := \{ d \in D \mid q \le |d| < 2q \}$.

In the algorithm we want to compute the set $T$ of all indices $i$ such that $s[1..i]$ can be partitioned (where $0 \in T$, since the empty string can be partitioned). The trivial $O(nm)$ algorithm computes $T \cap \{0,\ldots,i\}$ one by one, by checking for each $i$ whether for some string $d$ in the dictionary we have $s[i-|d|+1..i] = d$ and $i-|d| \in T$, since then we can extend the existing partitioning of $s[1..i-|d|]$ by the string $d$ to a partitioning of $s$. 

In our algorithm, when we have computed the set $T \cap \{0,\ldots,x\}$, we want to compute all possible ``jumps'' from a point before $x$ to a point after $x$ using dictionary words with length in $[q,2q)$ (for any power of two $q$). This gives rise to the following query problem. 

\begin{lemma} \label{lem:WBsubproblem}
On dictionary $D$ and string $s$, consider the following queries:

\begin{itemize}
\item \emph{Jump-Query}: Given a power of two $q \ge 1$, an index $x$ in $s$, and a set $S \subseteq \{x-2q+1,\ldots,x\}$, compute the set of all $x < i \le x + 2q$ such that we can $D_q$-jump from some $j \in S$ to $i$.
%More formally, we want to compute a sequence $R = R[x+1..x+2L]$ with
%$$ R[i] = \bigvee_{\substack{x-2L < j \le x \\ L \le i - j < 2L}} B[j] \wedge \big( s[j+1..i] \in D \big). $$
%(Here, $s[k]$ for $k \not \in \{1,\ldots,n\}$ is interpreted as a symbol not appearing in any dictionary word.)
\end{itemize}

We can preprocess $D,s$ in time $O(n \log m + m)$ such that queries of the above form can be answered in time $O(\min\{ q^2, \sqrt{q m \log q} \})$, where $m$ is the total size of~$D$ and $n = |s|$. 
\end{lemma}

Before we prove that jump-queries can be answered in the claimed running time, let us show that this implies an $\tilde O(n m^{1/3} + m)$-time algorithm for the Word Break problem.

\begin{proof}[Proof of Theorem~\ref{thm:wordbreakalgo}]
  The algorithm works as follows. After initializing $T := \{0\}$, we iterate over $x = 0,\ldots,n-1$. For any $x$, and any power of two $q \le n$ dividing $x$, define $S := T \cap \{x-2q+1,\ldots,x\}$. Solve a jump-query on $(q,x,S)$ to obtain a set $R \subseteq \{x+1..x+2q\}$, and set $T := T \cup R$. 
  
  To show correctness of the resulting set $T$, we have to show that $i \in \{0,\ldots,n\}$ is in $T$ if and only if $s[1..i]$ can be partitioned. Note that whenver we add $i$ to $T$ then $s[1..i]$ can be partitioned, since this only happens when there is a jump to $i$ from some $j \in T$, $j < i$, which inductively yields a partitioning of $s[1..i]$. For the other direction, we have to show that whenever $s[1..i]$ can be partitioned then we eventually add $i$ to $T$.
  This is trivially true for the empty string ($i = 0$). 
  For any $i >0$ such that $s[1..i]$ can be partitioned, consider any $0 \le j < i$ such that $s[1..j]$ can be partitioned and we can jump from $j$ to $i$. Round down $i-j$ to a power of two $q$, and consider any multiple $x$ of $q$ with $j \le x < i$. Inductively, we correctly have $j \in T$. Moreover, this holds already in iteration $x$, since after this time we only add indices larger than $x$ to $T$. Consider the jump-query for $q$, $x$, and $S := T \cap \{x-2q+1,\ldots,x\}$ in the above algorithm. In this query, we have $j \in S$ and we can jump from $j$ to $i$, so by correctness of Lemma~\ref{lem:WBsubproblem} the returned set $R$ contains~$i$. Hence, we add $i$ to $T$, and correctness follows.
  
  For the running time, since there are $O(n/q)$ multiples of $1 \le q \le n$ in $\{0,\ldots,n-1\}$, there are $O(n/q)$ invocations of the query algorithm with power of two $q \le n$. Thus, the total time of all queries is up to constant factors bounded by 
  $$ \sum_{i=0}^{\log n} \frac n{2^\ell} \cdot \min\Big\{ (2^\ell)^2, \sqrt{2^\ell m \log (2^\ell)} \Big\} 
  = n \cdot \sum_{i=0}^{\log n} \min\Big\{ 2^\ell, \sqrt{ m \ell / 2^\ell} \Big\}. $$
  We split the sum at a point $\ell^*$ where $2^{\ell^*} = \Theta((m \log m)^{1/3})$ and use the first term for smaller~$\ell$ and the second for larger. Using $\sum_{i=a}^b 2^i = O(2^b)$ and $\sum_{i=a}^b \sqrt{i / 2^i} = O(\sqrt{a / 2^a})$, we obtain the upper bound
  $$ \le n \cdot \sum_{i=0}^{\ell^*} 2^\ell + n\cdot \sum_{i=\ell^*+1}^{\log n} \sqrt{m \ell / 2^\ell} = O\Big(n 2^{\ell^*} + n \sqrt{ m \ell^* / 2^{\ell^*}} \Big) 
  = O\big( n (m \log m)^{1/3} \big), $$
  since $\ell^* = O(\log m)$ by choice of $2^{\ell^*} = \Theta((m \log m)^{1/3})$.
  Together with the preprocessing time $O(n \log m + m)$ of Lemma~\ref{lem:WBsubproblem}, we obtain the desired running time $O(n (m \log m)^{1/3} + m)$.
\end{proof}
 
It remains to design an algorithm for jump-queries. We present two methods, one with query time $O(q^2)$ and one with query time $O(\sqrt{q m \log q})$. The combined algorithm, where we first run the preprocessing of both methods, and then for each query run the method with the better guarantee on the query time, proves Lemma~\ref{lem:WBsubproblem}.

\subsection{Jump-Queries in Time $O(q^2)$}

The dictionary matching algorithm by Aho and Corasick~\cite{AC75} yields the following statement. 

\begin{lemma} \label{lem:ahocorasick}
  Given a set of strings $D'$, in time $O(\|D'\|)$ one can build a data structure allowing the following queries. Given a string $s'$ of length $n'$, we compute the set $Z$ of all substrings of $s'$ that are contained in $D'$, in time $O(n' + |Z|) \le O(n'^2)$. 
\end{lemma}

With this lemma, we design an algorithm for jump-queries as follows.
In the preprocessing, we simply build the data structure of the above lemma for each $D_q$, in total time $O(m)$.

For a jump-query $(q,x,S)$, we run the query of the above lemma on the substring $s[x-2q+1..x+2q]$ of~$s$. This yields all pairs $(j,i)$, $x-2q < j < i \le x+2q$, such that we can $D_q$-jump from $j$ to $i$. Iterating over these pairs and checking whether $j \in S$ gives a simple algorithm for solving the jump-query. The running time is $O(q^2)$, since the query of Lemma~\ref{lem:ahocorasick} runs in time quadratic in the length of the substring $s[x-2q+1..x+2q]$.

\subsection{Jump-Queries in Time $O(\sqrt{q m \log q})$}

The second algorithm for jump-queries is more involved. Note that if $q > m$ then $D_q = \emptyset$ and the jump-query is trivial. Hence, we may assume $q \le m$, in addition to $q \le n$.

\paragraph{Preprocessing.}
We denote the reverse of a string $d$ by $\rev{d}$, and let $\rev{D}_q := \{ \rev{d} \mid d \in D_q\}$.
We build a trie $\cT_q$ for each $\rev{D}_q$. Recall that a trie on a set of strings is a rooted tree with each edge labeled by an alphabet symbol, such that if we orient edges away from the root then no node has two outgoing edges with the same labels. We say that a node $v$ in the trie \emph{spells} the word that is formed by concatenating all symbols on the path from the root to $v$.
The set of strings spelled by the nodes in $\cT_q$ is exactly the set of all prefixes of strings in $\rev{D}_q$. Finally, we say that the nodes spelling strings in $\rev{D}_q$ are \emph{marked}.
%Some nodes are marked, in particular all leaves are marked. Finally, the set of strings spelled by the marked nodes equals $\rev{D}_q$ for trie $\cT_q$.
We further annotate the trie $\cT_q$ by storing for each node $v$ the lowest marked ancestor $m_v$. 

In the preprocessing we also run the algorithm of the following lemma.

\begin{lemma} \label{lem:stringdatastruct}
  The following problem can be solved in total time $O(n \log m + m)$. For each power of two $q \le \min\{n,m\}$ and each index $i$ in string $s$, compute the minimal $j=j(i)$ such that $s[j..i]$ is a suffix of a string in $D_q$. Furthermore, compute the node $v(q,i)$ in $\cT_q$ spelling the string $s[j(i)..i]$. 
\end{lemma}

Note that the second part of the problem is well-defined: $\cT_q$ stores the reversed strings $\rev{D}_q$, so for each suffix $x$ of a string in $D_q$ there is a node in $\cT_q$ spelling $x$.

\begin{proof}
  First note that the problem decomposes over $q$. Indeed, if we solve the problem for each~$q$ in time $O(\|D_q\| + n)$, then over all $q$ the total time is $O(m + n \log m)$, as the $D_q$ partition $D$ and there are $O(\log m)$ powers of two $q \le m$.
  
  Thus, fix a power of two $q \le \min\{n,m\}$. It is natural to reverse all involved strings, i.e., we instead want to compute for each $i$ the maximal $j$ such that $\rev{s}[i..j]$ is a prefix of a string in~$\rev{D}_q$. %Since the trie $\cT_q$ contains $\rev{D}_q$, any prefix of a string in $\rev{D}_q$ is spelled by some node in $\cT_q$, so the problem is well-defined. 
  
  Recall that a suffix tree is a compressed trie containing all suffixes of a given string $s'$. In particular, ``compressed'' means that if the trie would contain a path of degree 1 nodes, labeled by the symbols of a substring $s'[i..j]$, then this path is replaced by an edge, which is succinctly labeled by the pair $(i,j)$. 
  We call each node of the uncompressed trie a \emph{position} in the compressed trie, in other words, a position in a compressed trie is either one of its nodes or a pair $(e,k)$, where $e$ is one of the edges, labeled by $(i,j)$, and $i < k < j$. A position $p$ is an \emph{ancestor} of a position $p'$ if the corresponding nodes in the uncompressed tries have this relation, i.e., if we can reach $p$ from $p'$ by going up the compressed trie.
  It is well-known that suffix trees have linear size and can be computed in linear time~\cite{Weiner73}. In particular, iterating over all \emph{nodes} of a suffix tree takes linear time, while iterating over all \emph{positions} can take up to quadratic time (as each of the $n$ suffixes may give rise to $\Omega(n)$ positions on average).
  
  We compute a suffix tree $\cS$ of $\rev{s}$. 
  Now we determine for each node $v$ in $\cT_q$ the position $p_v$ in $\cS$ spelling the same string as $v$, if it exists. This task is easily solved by simultaneously traversing $\cT_q$ and $\cS$, for each edge in $\cT_q$ making a corresponding move in $\cS$, if possible. During this procedure, we store for each node in $\cS$ the corresponding node in $\cT_q$, if it exists. Moreover, for each edge $e$ in $\cS$ we store (if it exists) the pair $(v,k)$, where $k$ is the lowest position $(e,k)$ corresponding to some node in $\cT_q$, and $v$ is the corresponding node in $\cT_q$. Note that this procedure runs in time $O(\|D_q\|)$, as we can charge all operations to nodes in $\cT_q$.
  
  Since $\cS$ is a suffix tree of $\rev{s}$, each leaf $u$ of $\cS$ corresponds to some suffix $\rev{s}[i..n]$ of $\rev{s}$. With the above annotations of $\cS$, iterating over all nodes in $\cS$ we can determine for each leaf $u$ the lowest ancestor position $p$ of $u$ corresponding to some node $v$ in $\cT_q$. It is easy to see that the string spelled by $v$ is the longest prefix shared by $\rev{s}[i..n]$ and any string in $\rev{D}_q$. In other words, denoting by $\ell$ the length of the string spelled by $v$ (which is the depth of $v$ in $\cT_q$), the index $j := i+\ell-1$ is maximal such that $\rev{s}[i..j]$ is a prefix of a string in $\rev{D}_q$. Undoing the reversing, $j' := n+1-j$ is minimal such that $s[j'..n+1-i]$ is a suffix of a string in $D_q$. Hence, setting $v(q,n+1-i) := v$ solves the problem. 
  
  This second part of this algorithm performs one iteration over all nodes in $\cS$, taking time $O(n)$, while we charged the first part to the nodes in $\cT_q$, taking time linear in the size of $D_q$. In total over all $q$, we thus obtain the desired running time $O(n \log m + m)$.
%
%  Finally, for any leaf $u$ in $\cS$
%  for each position $p$ in $\cS$ corresponding to node $v$ in $\cT_q$, and each child $p'$ of $p$ not corresponding to any node in $\cT_q$, it is easy to see that 
%  
%  Finally, annotate for each $v$ in $\cT_q$ the corresponding position $p_v$ in $\cS$ by $v$. Consider an unannotated position $p$ in $\cS$ such that the preceding position (going up the tree by one symbol) is annotated by $v$. Each leaf of $\cS$ in the subtree below $p$ corresponds to some suffix $\rev{s}[i..n]$ of $\rev{s}$. Observe that the string spelled by $v$ in $\cT_q$ is the longest prefix shared by $\rev{s}[i..n]$ and any string in $\rev{D}_q$. In other words, denoting by $\ell$ the length of the string spelled by $v$, the index $j := i+\ell-1$ is maximal such that $\rev{s}[i..j]$ is a prefix of a string in $\rev{D}_q$. Undoing the reversing, $j' := n+1-j$ is minimal such that $s[j'..n+1-i]$ is a suffix of a string in $D_q$. Hence, setting $v(q,n+1-i) := v$ solves the problem. Note that this can be performed for all $i$ by one traversal over $\cS$.
%  
%  The time complexity is linear in the size of $\cT_q$ and $\cS$. Over all $q$, we thus obtain the desired running time $O(n \log m + m)$.
\end{proof}

For each $\cT_q$, we also compute a maximal packing of paths with many marked nodes, as is made precise in the following lemma. Recall that in the trie $\cT'$ for dictionary $D'$ the marked nodes are the ones spelling the strings in $D'$.

\begin{lemma}
  Given any trie $\cT$ and a parameter $\lambda$, a \emph{$\lambda$-packing} is a family $\cB$ of pairwise disjoint subsets of $V(\cT)$ such that (1) each $B \in \cB$ is a directed path in $\cT$, i.e., it is a path from some node $r_B$ to some descendant $v_B$ of $r_B$, (2) $r_B$ and $v_B$ are marked for any $B \in \cB$, and (3) each $B \in \cB$ contains exactly $\lambda$ marked nodes. 
  
  In time $O(|V(\cT)|)$ we can compute a maximal (i.e., non-extendable) $\lambda$-packing.
\end{lemma}
\begin{proof}
  We initialize $\cB = \emptyset$.
  We perform a depth first search on $\cT$, remembering the number $\ell_v$ of marked nodes on the path from the root to the current node $v$. When $v$ is a leaf and $\ell_v < \lambda$, then $v$ is not contained in any directed path containing $\lambda$ marked nodes, so we can backtrack. When we reach a node $v$ with $\ell_v = \lambda$, then from the path from the root to $v$ we delete the (possibly empty) prefix of unmarked nodes to obtain a new set $B$ that we add to $\cB$. Then we restart the algorithm on all unvisited subtrees of the path from the root to $v$. Correctness is immediate.
\end{proof}

For any power of two $q \le \min\{n,m\}$, we set $\lambda_q := \big(\frac mq \log q\big)^{1/2}$ and compute a $\lambda_q$-packing $\cB_q$ of $\cT_q$, in total time $O(m)$. In $\cT_q$, we annotate the highest node $r_B$ of each path $B \in \cB$ as being the \emph{root} of $B$. This concludes the preprocessing.

\paragraph{Query Algorithm.}
Consider a jump-query $(q,x,S)$ as in Lemma~\ref{lem:WBsubproblem}. For any $B \in \cB$ let $d_B$ be the string spelled by the root $r_B$ of $B$ in $\cT_q$, and let $\pi_B = (u_1,\ldots,u_k)$ be the path from the root of $\cT$ to the root $r_B$ of $B$ (note that the labels of $\pi_B$ form $d_B$). We set $S_B := \{ 1 \le i \le k \mid u_i \text{ is marked} \}$, which is the set containing the length of any prefix of $d_B$ that is contained in $\rev{D}_q$, as the marked nodes in $\cT_q$ correspond to the strings in $\rev{D}_q$.

As the first part of the query algorithm, we compute the sumsets $S+S_B := \{ i + j \mid i \in S,\, j \in S_B \}$ for all $B \in \cB$. 

Now consider any $x < i \le x+2q$.
By the preprocessing (Lemma~\ref{lem:stringdatastruct}), we know the minimal~$j$ such that $s[j..i]$ is a suffix of some $d \in D_q$, and we know the node $v := v(q,i)$ in $\cT_q$ spelling $s[j..i]$. Observe that the path $\sigma$ from the root to $v$ in $\cT_q$ spells the reverse of $s[j..i]$. It follows that the strings $d \in D_q$ such that $s[i-|d|+1..i] = d$ correspond to the marked nodes on $\sigma$. To solve the jump-query (for $i$) it would thus be sufficient to check for each marked node $u$ on $\sigma$ whether for the depth $j$ of $u$ we have $i - j \in S$, as then we can $D_q$-jump from $i-j$ to $j$ and have $i-j \in S$. Note that we can efficiently enumerate the marked nodes on $\sigma$, since each node in $\cT_q$ is annotated with its lowest marked ancestor. However, there may be up to $\Omega(q)$ marked nodes on $\sigma$, so this method would again result in running time $\Theta(q)$ for each $i$, or $\Theta(q^2)$ in total. 

Hence, we change this procedure as follows. Starting in $v = v(q,i)$, we repeatedly go the lowest marked ancestor and check whether it gives rise to a partitioning of $s[1..i]$, until we reach the root $r_B$ of some $B \in \cB$. Note that by maximality of $\cB$ we can visit less than $\lambda_q$ marked ancestors before we meet any node of some $B \in \cB$, and it takes less than $\lambda_q$ more steps to lowest marked ancestors to reach the root $r_B$. Thus, this part of the query algorithm takes time $O(\lambda_q)$. Observe that the remainder of the path $\sigma$ equals $\pi_B$. We thus can make use of the sumset $S + S_B$ as follows. The sumset $S + S_B$ contains $i$ if and only if for some $1 \le j \le |\pi_B|$ we have $i-j \in S$ and we can $D_q$-jump from $i-j$ to $i$. Hence, we simply need to check whether $i \in S + S_B$ to finish the jump-query for $i$.

\paragraph{Running Time.}
As argued above, the second part of the query algorithm takes time $O(\lambda_q)$ for each $i$, which yields $O(q \cdot \lambda_q)$ in total. 

For the first part of computing the sumsets, first note that $D_q$ contains at most $m/q$ strings, since its total size is at most $m$ and each string has length at least $q$. Thus, the total number of marked nodes in $\cT_q$ is at most $m/q$. As each $B \in \cB$ contains exactly $\lambda_q$ marked nodes, we have
\begin{equation} \label{eq:cB}
  |\cB| \le m / (q \cdot \lambda_q).
\end{equation} 
For each $B \in \cB$ we compute a sumset $S + S_B$.
Note that $S$ and $S_B$ both live in universes of size $O(q)$, since $S \subseteq \{x-2q+1,\ldots,x\}$ by definition of jump-queries, and all strings in $D_q$ have length less than $2q$ and thus $|S_B| \subseteq \{1,\ldots,2q\}$. After translation, we can even assume that $S, S_B \subseteq \{1,\ldots,O(q)\}$. 
It is well-known that computing the sumset of $X,Y \subseteq \{1,\ldots,U\}$ is equivalent to computing the Boolean convolution of their indicator vectors of length $U$. The latter in turn can be reduced to multiplication of $O(U \log U)$-bit numbers, by padding every bit of an indicator vector with $O(\log U)$ zero bits and concatenating all padded bits. Since multiplication is in linear time on the Word RAM, this yields an $O(U \log U)$ algorithm for sumset computation.
Hence, performing a sumset computation $S + S_B$ can be performed in time $O(q \log q)$. Over all $B \in \cB$, we obtain a running time of $O(|\cB| \cdot q \log q) = O((m \log q) / \lambda_q)$, by the bound~(\ref{eq:cB}). 

Summing up both parts of the query algorithm yields running time $O(q \cdot \lambda_q + (m \log q) / \lambda_q)$. Note that our choice of $\lambda_q = \big(\frac mq \log q\big)^{1/2}$ minimizes this time bound and yields the desired query time $O(\sqrt{q m \log q})$. This finishes the proof of Lemma~\ref{lem:WBsubproblem}.

% !TEX root = main.tex

\section{Almost-linear Time Algorithms}
\label{sec:almostlinear}

In this section we prove Theorem~\ref{thm:newalgos}, i.e. we present
an $\tilde{O}(n)+O(m)$ time algorithm for $\pipe \plus \circ \plus$-membership and an
$n^{1+o(1)}+O(m)$ time algorithm for $\pipe \plus \circ
\pipe$-membership. We start with presenting the solution for $\pipe
\plus \circ \pipe$-membership as many of the ideas carry over to
$\pipe \plus \circ \plus$-membership.

\subsection{Almost-linear Time for $\pipe \plus \circ \pipe$}
\label{sec:firstUpper}
For a given length-$n$ string $T$ and length-$m$ regular expression $R$
of type $\pipe \plus \circ \pipe$, over an alphabet $\Sigma$, let
$R_1,\dots,R_k$ denote the regular expressions of type $\circ
\pipe$ such that $R=R_1^\plus \pipe R_2^\plus
\pipe \cdots \pipe R_k^\plus \pipe \sigma_1 \pipe \cdots \pipe
\sigma_j$. Here the $\sigma_j$'s are characters from $\Sigma$
(recall that in the definition of homogenous regular expressions we
allow leaves in any depth, so we can
have the single characters $\sigma_i$ in $R$). Since the $\sigma_i$'s are
trivial to handle, we ignore them in the remainder.

For convenience, we index the
characters of $T$ by $T[0],\dots,T[n-1]$. For $R$ to match $T$, it must be the
case that $R_i^\plus$ matches $T$ for some index $i$. Letting $\ell_i$
be the number of $\circ$'s in $R_i$, we define $S_{i,j}
\subseteq \Sigma$ for $j=0,\dots,\ell_i$ as the set of characters from $\Sigma$ such that
$$
R_i = (\pipe_{\sigma \in S_{i,0}} \sigma) \circ (\pipe_{\sigma \in
  S_{i,1}} \sigma) \circ \cdots \circ (\pipe_{\sigma \in S_{i,\ell_i}}
\sigma).
$$
Note that if a leaf appears in the $\pipe$-level, then 
the set $S_{i,j}$ is simply a singleton set.

We observe that $T$ matches $R_i^\plus$ iff $(\ell_i+1)$ divides $|T|=n$
and $T[j] \in S_{i, j \mmod (\ell_i+1)}$ for all
$j=0,\dots,n-1$. In other words, if $(\ell_i+1)$ divides $n$ and we define sets $T^{\ell_i}_j \subseteq \Sigma$
for $j=0,\dots,\ell_i$, such that 
$$
T^{\ell_i}_j = \bigcup_{h=0}^{n/(\ell_i+1)-1}
\{T[h(\ell_i+1)+j]\},$$
then we see that $T$ matches $R_i^\plus$ iff $T_j
\subseteq S_{i, j}$ for $j=0,\dots,\ell_i$.

Note that the sets $T^{\ell_i}_j$ depend only on $T$ and $\ell_i$,
i.e. the number of $\circ$'s in $R_i$.  We therefore
start by partitioning the expressions $R_i$ into groups having the
same number of $\circ$'s $\ell = \ell_i$. This takes time $O(m)$. We can immediately discard all groups
where $(\ell+1)$ does not divide $n$. The crucial property we will use
is that an integer $n$ can have no more than $2^{O(\lg n/\lg \lg n)}$
distinct divisors~\cite{WI1906}, so we have to consider at most $2^{O(\lg n/\lg \lg n)}$ groups. 

Now let $R_{i_1},\dots,R_{i_k}$ be the regular expressions in a group,
i.e., $\ell = \ell_{i_1}=\ell_{i_2}=\cdots =\ell_{i_k}$. By a linear scan
through $T$, we compute in $O(n)$ time the sets $T_j^{\ell}$ for
$j=0,\dots,\ell$. We store the sets in a hash table for expected
constant time lookups, and we store the sizes $|T^\ell_j|$. We then check whether
there is an $R_{i_h}$ such that $T^\ell_j \subseteq S_{i_h, j}$ for
all $j$. This is done by examining each $R_{i_h}$ in turn. For each
such expression, we check whether $T^\ell_j \subseteq S_{i_h, j}$ for
all $j$. For one $S_{i_h,j}$, this is done by
taking each character of $S_{i_h,j}$ and testing for membership in
$T^\ell_j$. From this, we can compute $|T^\ell_j \cap S_{i_h,j}|$. We conclude that $T^\ell_j \subseteq S_{i_h, j}$ iff $|T^\ell_j \cap S_{i_h,j}|=|T^\ell_j|$. 

All the membership testings, summed over the entire execution of the
algorithm, take expected $O(m)$ time as we make at most one query per
symbol of the input regular expression. Computing the sets $T^\ell_j$
for each divisor $(\ell+1)$ of $n$ takes $n2^{O(\lg n/\lg \lg
  n)}$ time. Thus, we conclude that $\pipe \plus \circ \pipe$-membership
testing can be solved in expected time $n^{1+o(1)}+O(m)$.

\paragraph{Sub-types.}
We argue that the above algorithm also solves any type $t$ where $t$
is a subsequence of $\pipe \plus \circ \pipe$. Type $\plus
\circ \pipe$ simply corresponds to the case of just one $R_i$ and is
thus handled by our algorithm above. Moreover, since there is only one
$R_i$ and thus only one divisor $\ell_i+1$, the running time of our algorithm
improves to $O(n+m)$. Type $\pipe \circ \pipe$ can be
solved by first discarding all $R_i$ with $\ell_i \neq n-1$ and then
running the above algorithm. Again this leaves only one value of
$\ell_i$ and thus the above algorithm runs in time $O(n+m)$. The type $\pipe \plus \pipe$ corresponds to
the case where each $\ell_i=0$ and is thus also handled by the above
algorithm. Again the running time becomes $O(n+m)$ as there is only
one value of $\ell_i$. Type $\pipe \plus \circ$ is the case where all
sets $S_{i,j}$ are singleton sets and is thus also handled by the
above algorithm. However, this type is also a subsequence of $\pipe \plus
\circ \plus$ and using the algorithm developed in the next section, we
get a faster algorithm for $\pipe \plus \circ$ than using the one above.

Type $\pipe \pipe$, $\pipe \circ$, $\pipe \plus$ are trivial. Type
$\plus \pipe$ corresponds to the case of just one $R_i$ having
$\ell_i=0$ and is thus solved in $O(n+m)$ time using our
algorithm. Type $\plus \circ$ corresponds to just one $R_i$ and only
singleton sets $S_{i,j}$ and thus is also solved in $O(n+m)$ time by
the above algorithm. The type $\circ \pipe$ is the special case of
$\pipe \circ \pipe$ in which there is only one set $R_i$ and is thus
also solved in $O(n+m)$ time. Types with just one operator are trivial.

\subsection{Near-linear Time for $\pipe \plus \circ \plus$}
For a given length-$n$ text $T$ and length-$m$ regular expression $R$
of type $\pipe \plus \circ \plus$, over an alphabet $\Sigma$, let
$R_1,\dots,R_k$ denote the regular expressions of type $\circ
\plus$ such that $R=R_1^\plus \pipe R_2^\plus
\pipe \cdots \pipe R_k^\plus \pipe \sigma_1 \pipe \cdots \pipe
\sigma_j$. As in Section~\ref{sec:firstUpper}, the $\sigma_i$'s are
single characters. These can easily be tested against $T$ and thus
from now on we ignore them.

Our new algorithm uses some of the ideas from Backurs and Indyk~\cite{backursindyk} for
$\plus \circ \plus$-membership. From the text $T$, define its
run-length encoding $r(T)$ as follows: Set $T' = T$ and let $r(T)$ be
an initially empty list of tuples. While $|T'|>0$, let $\sigma$ be the first
symbol of $T'$ and let $\ell > 0$ be the largest integer such that
$\sigma^\ell$ is a prefix of $T'$ (i.e. $T'$ starts with $\ell$
$\sigma$'s). We remove the prefix $\sigma^\ell$ from $T'$ and append the
tuple $(\sigma, \ell)$ to $r(T)$.

Following Backurs and Indyk~\cite{backursindyk}, we also define the run-length
encoding of a regular expression $R_i$ of type $\circ \plus$ as
follows: Let $R_i'=R_i$ and let $r(R_i)$ be an initially empty
sequence of tuples. While $|R_i'| > 0$, let $\ell>0$ be the largest
integer such that there exists a length-$\ell$ prefix of $R_i'$ of the form $\sigma
\sigma^\plus \sigma^\plus \sigma \sigma \dots$ (an arbitrary
concatenation of $\sigma$ and $\sigma^\plus$) for a symbol $\sigma \in
\Sigma$. Define $\ell' \geq 0$ as the number of $\sigma$'s in the
prefix (and $\ell-\ell'$ is the number of $\sigma^\plus$'s). If
$\ell'=\ell$, we append the tuple $(\sigma, =\ell)$ to
$r(R_i)$. Otherwise, we append the tuple $(\sigma, \geq \ell)$ to
$r(R_i)$. We then delete the prefix and repeat until $R_i'$ has length
$0$.

Backurs and Indyk observed the following for matching $T$ and
$R_i^\plus$ for an $R_i$
of type $\circ \plus$: If the first and last character of $R_i$ are
distinct, then $T$ matches $R_i^\plus$
iff the following two things hold:
\begin{enumerate}
\item $|r(R_i)|$ divides $|r(T)|$.
\item For every $j=0,\dots,|r(T)|-1$, if
$(\sigma, \ell)$ denotes the $j$'th tuple of $r(T)$, we must have
that the $(j \mmod |r(R_i)|)$'th tuple of $r(R_i)$ (counting from $0$) is either of the
form $(\sigma, =\ell)$ or $(\sigma, \geq \ell')$ for some $\ell' \leq
\ell$. 
\end{enumerate}

The case where the first and last character of $R_i$ are the same can be efficiently reduced to the case of distinct
characters. We argue how at the end of this section and now proceed
under the assumption that each $R_i$ has distinct first and last
character.

Compared to the $\plus \circ \plus$-case solved by Backurs and Indyk,
we need to additionally handle an outer $\pipe$. Our solution for $\pipe
\plus \circ \pipe$ hints at how: We start by partitioning the $R_i$'s
into groups such that all $R_i$'s with the same value of $|r(R_i)|$
are placed in the same group. This can easily be done in $O(m)$
time. As argued in Section~\ref{sec:firstUpper}, there are at most
$2^{O(\lg n/\lg \lg n)}$ groups we need to care about, as property~1 above implies that $|r(R_i)|$ must divide $|r(T)|$ for $T$ to
possibly match $R_i^\plus$.

For a length $s$ dividing $|r(T)|$, we check $T$ against all the $R_i^\plus$'s with
$|r(R_i)|=s$ as follows: First let $(\sigma_j , \ell_j)$ denote the
$j$'th tuple in $r(T)$ for $j=0,\dots,|r(T)|-1$. If there are two
tuples $(\sigma_j, \ell_j)$ and $(\sigma_h, \ell_h)$ with $\sigma_j
\neq \sigma_h$ and $j \mmod s = h \mmod s$, we can conclude that $T$
cannot possible match $R_i^\plus$ for any $R_i$ in the group (due to
property 2 above). If this
is not the case, assume for now that we have somehow computed the following
values $\alpha^s_j$ and $\beta^s_j$ for $j=0,\dots,s-1$:
$$
\alpha^s_j := \min_{h=0}^{|r(T)|/s - 1} \ell_{hs+j}.
$$
$$
\beta^s_j := \max_{h=0}^{|r(T)|/s - 1} \ell_{hs+j}.
$$
Also, let $\sigma_j \in \Sigma$ denote the character such that
$\sigma_h = \sigma_j$ for all $h$ satisfying $h \mmod s = j$.

We then test $T$ against each $R_i^\plus$ as follows: For
$j=0,\dots,s-1$, consider the $j$'th tuple in
$r(R_i)$. We have two cases:
\begin{enumerate}[a]
\item If the $j$'th tuple is of the form $(\mu, =\ell)$, we check
whether $\mu=\sigma_j$. If not, we conclude that $T$ and $R_i^\plus$
cannot match. Otherwise, we check whether $\alpha^s_j=\beta^s_j =
\ell$. If not, we also conclude that $T$ and $R_i^\plus$ cannot
match.
\item If the $j$'th tuple is of the form $(\mu, \geq \ell)$,
we check whether $\mu = \sigma_j$. If not, we conclude that $T$ and
$R_i^\plus$ cannot match. Otherwise, we check whether $\ell \leq
\alpha^s_j$. If not, we conclude that $T$ and $R_i^\plus$ cannot match.
\end{enumerate}
It follows from property 2 above that $T$ and $R_i^\plus$ do not match if and only if
the above procedure concludes that they do not match.

Assuming the availability of $\alpha^s_j$ and $\beta^s_j$, testing $T$
against an $R_i^\plus$ thus takes time $O(|r(R_i)|)$. Summing over all
$R_i$, this is $O(m)$ in total. Thus we only need an efficient
algorithm for computing the $\alpha^s_j$'s and $\beta^s_j$'s. We could
compute them in $O(n)$ time per divisor $s$, resulting in an algorithm
with running time $n2^{O(\lg n/\lg \lg n)}+O(m)$ as in
Section~\ref{sec:firstUpper}. However, for this problem we can do
better. To compute the $\alpha^s_j$'s and $\beta^s_j$'s for all
divisors $s$ of $|r(T)|$, we start by forming a tree over all the divisors
of $|r(T)|$. This tree is formed by first computing all $2^{O(\lg n/\lg \lg
  n)}$ divisors of $|r(T)|$ in $O(n)$ time. For each divisor $s$, we
compute a prime factorization of $s$ in time $O(s)$. We let the
divisor $|r(T)|$ be the root of the tree, and each divisor $s < |r(T)|$ is
assigned as a child of the smallest divisor $t$ of $|r(T)|$ such that $s$
divides $t$. Note that the smallest such divisor $t$ can be found in
$O(\lg n/\lg \lg n)$ time from the prime factorization of $|r(T)|$ and $s$ (simply
multiply $s$ by the smallest prime that occurs more times in the
factorization of $|r(T)|$ than in $s$ and note that $|r(T)|$ has at most $O(\lg
n/ \lg \lg n)$ distinct prime factors). We compute the values $\alpha^s_j$ and
$\beta^s_j$ for all $s$ by a top-down sweep of the constructed
tree. We start at the root divisor $|r(T)|$ where we compute the
$\alpha_j^{|r(T)|}$'s and $\beta_j^{|r(T)|}$'s trivially in $O(n)$ time. When
processing a divisor $s$, let $t$ be its parent and let $p$ be the
prime such that $p = t/s$. We observe that
$$
\alpha^s_j = \min_{h=0}^{|r(T)|/s - 1} \ell_{hs+j} = \min_{q=0}^{p-1}
\min_{h=0}^{|r(T)|/t-1} \ell_{ht+qs+j} =  \min_{q=0}^{p-1} \alpha^t_{qs+j}
$$
and
$$
\beta^t_j = \max_{q=0}^{p-1} \beta^t_{qs+j}.
$$
Thus for a fixed $s$, we can compute all these values in time $O(sp) =
O(t)$ given access to the $\alpha^t_j$'s where $t$ is the parent of
$s$. Now each node of the tree has at most $O(\lg n/ \lg \lg n)$
children, as this is the maximum number of distinct prime factors of
an integer no greater than $n$. We thus conclude that the total time
for computing all the $\alpha^s_j$'s over all the divisors $s$ is at
most $O(\lg n/\lg \lg n) \cdot \sum_{s : s \textrm{ divides } |r(T)|}
s$. The sum over all divisors of an integer $n$ is known to be $O(n \lg \lg n)$~\cite{Gronwall},
so we conclude that the total time for computing the $\alpha^s_j$'s
and $\beta^s_j$'s is $O(n \lg n) = \tilde{O}(n)$.

Note that we also need to compute for each divisor $s$ of $|r(T)|$ if
there are two distinct characters $\sigma_j \neq \sigma_h$ in tuples
$(\sigma_j, \ell_j)$ and $(\sigma_h, \ell_h)$ with $j \mmod s = h \mmod
s$. If this was the case, we knew that $T$ couldn't possibly match any
of the $R_i^\plus$'s with $|r(R_i)|=s$. Observing that if $s$ divides
$t$ and we know that two such characters exist for the divisor $t$,
then this is also the case for $s$. Thus as for the
$\alpha_j^s$'s and $\beta^s_j$'s, we can compute this for all divisors
using a top-down sweep
of the tree in time $O(n \lg n)$. We conclude that our
algorithm runs in time $O(n \lg n + m) = \tilde{O}(n)+O(m)$.

\paragraph{Handling Identical First and Last Characters.}
Given an input where some of the $R_i$'s start and end with the
same character, we extract all these $R_i$'s. 
For each $R_i$, we first
check if $R_i$ and $T$ both end and start with the same character. If
not, we conclude that $T$ cannot possible match $R_i^\plus$ and thus
we can discard $R_i$. Next, if all characters of $R_i$ are the same ($|r(R_i)|=1$), testing
it against $T$ is trivial (we can precompute whether $T$ has only one
character, and thus we can test for matching in $O(|R_i|)$ time for each
such $R_i$). 

The remaining $R_i$'s start and end with the same character as
$T$. Let $r(R_i)$ be the run-length encoding of $R_i$ and let $r(T)$
be the run-length encoding of $T$. We check whether the last tuple of 
$T$ can possible match the last tuple of $r(R_i)$: If $(\sigma,
\ell)$ is the last tuple of $T$, we check whether the last tuple of
$r(R_i)$ is either $(\sigma, =\ell)$ or $(\sigma, \geq \ell')$ for an
$\ell' \leq \ell$. If not, we can discard $R_i$. We make the same test
with the first tuple of $r(R_i)$ and $r(T)$. 

We now ``rotate'' the text and the remaining patterns $R_i$ as
follows: Let $\sigma$ be the first and last character of $T$ and all
the remaining $R_i$'s. We take all occurrences of $\sigma$ at the end
of $T$ and move them to the front (think of it as a cyclic rotation of
$T$). We do the same thing with the $R_i$'s, i.e. take all occurrences
of $\sigma$ and $\sigma^+$ at the end of $R_i$ and move to the
front. The observation is that the ``rotated'' $T$ matches a ``rotated''
remaining pattern $R_i$ iff they matched before the
rotation. Moreover, $T$ and all the $R_i$'s now start with $\sigma$
and end with something else. Thus we have reduced to the
case of distinct first and last characters. The reduction took time
$O(n+m)$.

\paragraph{Sub-types.}
Type $\plus \circ \plus$ corresponds to the case of only one $R_i$. As
this means there is only one size $s=|r(R_i)|$, we can compute the
$\alpha^s_j$ and $\beta^s_j$ values directly in $O(n)$ time, giving
total time $O(n+m)$. Type $\pipe \circ \plus$ is solved by first
discarding all $R_i$ where $|r(R_i)| \neq |r(T)|$ (assuming we have
already reduced to the case of distinct first and last symbols). This
similarly leaves just one divisor of $|r(T)|$ to compute $\alpha_j^s$
and $\beta_j^s$ values for and thus the
running time improves to $O(n+m)$. Type $\pipe \plus \plus$ is
trivial. Type $\pipe \plus \circ$ corresponds to the case where there
are no
tuples $(\sigma, \geq \ell)$ in any $r(R_i)$ (i.e. all tuples are of
the form $(\sigma, =\ell)$). We thus get $\tilde{O}(n)+O(m)$ for this
case by using the above algorithm. Using the more directly taylored
algorithm of Backurs and Indyk~\cite{backursindyk}, this can be reduced to
$O(n+m)$. 

The only length two types not covered by the algorithm in
Section~\ref{sec:firstUpper} are $\plus \plus$ and $\circ \plus$. Type
$\plus \plus$ is trivial. Type $\circ \plus$ is the type $\pipe \circ
\plus$ with just one set $R_i$. It is thus solved in time $O(n+m)$ by
the above algorithm.

% !TEX root = main.tex

\section{SETH-Based Lower Bounds}
\label{sec:sethlower}

In this section we prove SETH-based lower bounds for $t$-membership
for types $\plus \pipe \circ \plus$, $\plus \pipe \circ \pipe$, and
$\pipe \plus \pipe \circ$.

All three proofs are reductions from Orthogonal Vectors and follow the
same overall approach. Let $A, B$ be two sets of input vectors for the
Orthogonal Vectors problem of size $n$ and $m$ respectively.  Each
reduction makes the set of vectors $A$ into a string and the set of
vectors $B$ into a regular expression of the considered type, such
that the string constructed from $A$ matches the regular expression
constructed from $B$ iff the Orthogonal Vectors instance has an
orthogonal pair of vectors.

The idea is to create a regular expression for each vector $b\in B$
that is matched by strings that encode vectors that are orthogonal to
$b$. The string for the vectors in $A$ are concatenated together in
such a way that the regular expression can test for each vector in
$B$, if any vector in $A$ is orthogonal to it.  The \textit{OR} of
these checks is implemented in the the string and regular expression
by the offset construction described in the $k$-Clique reduction in
Section~\ref{sec:wb_lower}. The offset construction is implemented
with the regular expression $\plus \pipe$ that is a part of all three
considered regular expression types. As mentioned earlier the
$\plus \pipe$ type can implement a dictionary of regular expression to
match against. The $\pipe$ allows to pick a regular expression from
the dictionary (subtree) for the input string to match, and the
$\plus$ allows doing that the repeatedly, matching substrings from the
input string to dictionary elements from left to right.
The dictionary we construct contains for each $b\in B$ a regular
expression that is matched only by orthogonal vectors and then 
regular expressions for the offset construction.

We directly apply the notation from the reduction in
Section~\ref{sec:wb_lower} as needed.
\begin{theorem} \label{thm:hardnessone}
  $\plus \pipe \circ \pipe$-membership takes time $(nm)^{1-o(1)}$ unless SETH fails.
\end{theorem}
\begin{proof}
  The string $T(A)$ is constructed as
  $$
  T(A) = \circ_{a\in A} [\mu a \$ \mu]^{(0)}_{\alpha,\beta} \circ \mu \mu 
  $$
  using the $\alpha,\beta$ encoding as defined
  Section~\ref{sec:wb_lower}. This is the concatenation of all vectors
  in $A$ with $\alpha,\beta$ surrounding each bit symbol, $\mu, \$
  \mu$ symbols around each encoded vector, and finally two $\mu$
  symbols at the end of the string.
  
  The regular expression must be of the form $(p_1\pipe p_2
  \pipe \dots \pipe p_t)^\plus$ where each $p_i$ is of type $\circ
  \pipe$ and these we construct the following way.
  %% There will be such a regular expression for each vector $b
  %% \in B$ and some additional regular expressions for implementing the
  %% offsets as in Section~\ref{sec:wb_lower}.
%  
  Let $c(1) = 0, c(0) = 0 \pipe 1$ encode simple regular expressions,
  and note that $c(1)$ is matched only by 0 while $c(0)$ is matched by
  both 0 and 1. For each $b\in B$ define $ C(b) = \circ_{i=1}^d
  c(b[i])$, the concatenation of the $c-$encoding of the bits in the
  vector. By construction, a vector $a \in A$ is orthogonal to the
  vector $b \in B$ iff $a$ read as a string matches the regular
  expression $C(b)$. Note that $C(b)$ is of the required form $\circ
  \pipe$.
  
  The full list $\mathcal{L}=p_1,\dots,p_t$ of patterns is defined as
  follows.  First, the regular expressions for checking orthogonality
  (notice the different offsets compared to the encoding of $A$ as a
  string)
$$
[C(b)]_{\alpha,\beta}^{(1)}, \quad \forall b\in B
$$
Next, add regular expressions that allows skipping a prefix of the $A$
vectors in the string $T(A)$ before matching an orthogonal vectors
pair, and the patterns that allows skipping vectors in the string
$T(A)$ after an orthogonal vectors pair has occured:
$$
\alpha \sigma \beta \textrm{ and } \beta \alpha \sigma \textrm{ for }\sigma \in \{0,1,\$,\mu\}.
$$
Then, we need the regular expressions for controlling the offset which are 
$$
\alpha \mu \beta \alpha, \$ \beta \alpha \mu, \beta \mu \mu,
$$
the regular expression that must be used before the start of an
orthogonal vectors match, the regular expression for changing offset
after a succesfull orthogonal vectors match, and finally the regular
expression that allows finishing the string match if an orthogonal
vector match regular expression has been succesfully used. Note
that all the regular expressions for implementing the offset are
concatenation of symbols (leafs) and thus off the right form.
The final regular expression becomes $R(B) = (\pipe_{\ell \in
  \mathcal{L}} \ell)^\plus$.

This finishes the reduction. There is a pair of orthogonal vectors
$a\in A,b \in B$ iff the string $T(A)$ matches the regular expression $R(B)$.
The proof of that is essentially the $K$-clique proof for Word-Break
from Section \ref{sec:wb_lower} and is omitted.
The string $T(A)$ has size $O(nd)$ and and the regular expression $R(B)$
has size $O(md)$ and the total construction time is linear
$O(d(n+m))$.
Thus, any algorithm for $\plus \pipe \circ \pipe$-membership that runs
in time $f(n,m)$ gives an algorithm for Orthogonal Vectors that runs
in time $O(d(n+m)+f(dn,dm))$.
In particular any algorithm for $\plus \pipe \circ \pipe$ membership
that runs in $O((nm)^{1-\eps})$ time for any $\eps>0$ gives an
algorithm for Orthogonal Vectors that runs in time $O((nd)^{2-2\eps})
= O(n^{2-2\eps} \poly(d))$ contradicting Conjecture~\ref{conj:ov}.

%% Assume there is a pair of vector $a_i,b_j$. We show how the string
%% $T(A)$ match the regular expression for $B$, by scanning the string
%% from left to right, repeatedly using a regular expression to move
%% right. For the prefix of the string that encodes
%% $a_1,\dots,a_{i-1}, move right by matching the each string with the
%% $\alpha \sigma \beta$ rule. This always matches complete coding of
%% $a_1,\dots,a_{i-1}.

%% Next move right by applying the start OV match gadget $\alpha \mu
%% \beta \alpha$, such that in the text the next thing to be matched
%% is $c(a_i[1]j) \beta \alpha \dots$. This part of $T$ defined by
%% $a_i$ is matched with $[C(b_i)]_{\alpha,\beta}^1$ since $a_i$ is
%% orthogonal to $b_j$, and the match ends having matched $a_i[d] \beta
%% \alpha$ and the following string symbols are $ beta \alpha \mu
%% \beta$. From here we match with the shift offset regular expression  $\beta
%% \alpha \mu$, and continue to match the remaining part of T(A)
%% coding $a_{i+1},\dots,a_{n}$  with the $\beta \alpha \sigma$
%% regular expressions that now matches everything in $T$ but the final three
%% symbols $\beta \mu \mu$. These three are of course matched with
%% corresponding end match regular expression.

%% The other way just reverses this scan, consdering and scanning
%% backwards from right to left in the string.  The only way to start
%% this match from the right side is to match the end regular
%% expression. shifts, which it must to be able to match the entire
%% string, and therea is the orthogonal vector pair. Formally, assume
%% text $T(A)$ matches pattern $P(B)$.

\end{proof}

\begin{theorem} \label{thm:hardnesstwo}
  $\plus \pipe \circ \plus$-membership takes time $(nm)^{1-o(1)}$ unless SETH fails.
\end{theorem}
\begin{proof}
  This reduction is very similar to the reduction above, the only
  difference is how the ortogonality check is encoded in the regular
  expressions now using $\circ \plus$ instead of $\circ \pipe$. This
  is achieved as follows.

  For each vector $a \in A$ construct a string where each zero bit is
  replaced with $t(0) = 001$, and each one bit is replaced with $t(1) =
  01$.
  Let $c(0) = 0^\plus 1, c(1) = 01$ and for each $b\in B$,
  define $C(b) = \circ_{i=1}^d c(b[i])$.
  An encoded zero in $b$ is matched by both $t(0)$ and $t(1)$, an
  encoded one is matched only by $t(0)$, and each may only be matched
  by the encoding of one bit from the string.  Thus, a vector $a$ is
  orthogonal to a vector $b$ iff the string encoding of vector $a$
  matches the regular expression $C(b)$.
  
  The remaining part of the reduction is identical to above and
  Theorem~\ref{thm:hardnesstwo} follows.
\end{proof}

\begin{theorem} \label{thm:hardnessthree}
  $\pipe \plus \pipe \circ$-membership takes time $(nm)^{1-o(1)}$ unless SETH fails.
\end{theorem}
\begin{proof}
  The regular expression we consider is now of the form $(p_1\pipe
  \dots \pipe p_m)$. In this reduction we construct regular
  expressions $p_i$ that is matched by the string constructed from the
  vectors in $A$ iff the $i$'th vector in $B$ is orthogonal to a
  vector in $A$.  The outer $\pipe$ functions acts as an \textit{OR},
  the following $\plus \pipe$ allows the offset implementation as
  above, and finally the concatenation operator $\circ$ is used to
  create regular expression that implements orthogonality checking.
  
  The orthogonality checks for vectors in $A,B$ are constructed as
  follows. For each $a = (a[1],\dots,a[d])$ we encode the position in the
  vector for each bit into a string as follows:
  $$
  t(a) =\#_1a[1] \dots \#_da[d],
  $$
  where $\#_{[1-d]}$ are $d$ unique symbols.

  The string $T(A)$ is defined as
  $$
  T(A) = \circ_{a\in A} [\mu t(a) \$ \mu]^{(0)}_{\alpha,\beta} \circ
  \mu \mu
  $$
  the concatenation of the encoding of each vector in $A$, again
  encoded for the offset construction.
  
  The regular expression for each $b \in B$ is constructed by making a
  list $\mathcal{L}$ of regular expressions of the from $\pipe \circ$
  which is then used to make the regular expression $R(b) =
  (\pipe_{p\in \mathcal{L}} p)^\plus$.
  The list $\mathcal{L}$ for $b$ is defined as follows. For each bit
  $b[i]$, we add the regular expression $\#_i 0$ to $\mathcal{L}$
  and if $b[i]$ is equal to zero we also add the regular expression
  $\#_i 1$. Denote that list of expressions $\mathcal{L}_{b}$. Then by
  construction vectors $a$ and $b$ are orthgonal if and only if $t(a)$
  matches the regular expression $(\pipe_{p\in \mathcal{L}_{b_i}}
  p)^\plus$.

  To finalize the construction, we replace the regular expressions in
  $\mathcal{L}_{b_i}$ with their $\alpha,\beta$ encodings and add the
  offset regular expression exactly as above to get $\mathcal{L}$.
  This construction is used for each input vector $b \in B$, and
  combined with $\pipe$ to get the  regular expression for $R(B) = (\pipe_{b\in B} R(b))$.
  %For completeness the construction is...

  The correctness arguments remain the same.  The string $T(A)$ only
  matches the regular expression $R(b)$ if there is a vector $a\in A$
  that is orthogonal to $b$. The top level $\pipe$ for the regular
  expression then allows to pick any $b$ that is orthogonal to a
  vector in $A$.
  Thus, if there is a pair of orthogonal vectors $a\in A,b \in B$ the
  string $T(A)$ matches the regular expression $R(b)$, and thus the
  regular expression for $B$.  If there is no orthogonal pair of
  vectors in $A,B$ then no $R(b)$ is matched by $A$, and thus the
  regular expression for $B$, the union of them, is not matched by
  $A$.

  The reduction constructs a string linear in the size of $A$, and a
  regular expression linear in the size of $B$ in total time linear in
  the size of $A$ and $B$. Thus, for the same reasons as above,
  Theorem \ref{thm:hardnessthree} follows
\end{proof}

% !TEX root = main.tex

\section{Dichotomy}
\label{sec:dichotomy}
In this section we prove Theorem~\ref{thm:main}, i.e., we show that the remaining results in this paper yield a complete dichotomy. 

%\begin{lemma} \label{lem:simplify}
%  For any type $t$, applying any of the following rules yields a type $t'$ such that $t$-membership and $t'$-membership are equivalent under linear-time reductions:
%  \begin{enumerate}
%    \item replace any substring $p p$, for any $p \in \{\circ,\pipe,\star,\plus\}$, by $p$,
%    %\item replace any substring $\star \plus$ by $\star$,
%    %\item replace any substring $\star \pipe \plus$ (or $\star \pipe \star$) by $\star \pipe$,
%    \item replace any substring $\plus \pipe \plus$ by $\plus \pipe$,
%    \item replace prefix $r \star$ by $r \plus$ for any $r \in \{\plus,\pipe\}^*$.
%  \end{enumerate}
%  We say that a $t$-membership \emph{simplifies} if one of these rules applies.
%  Applying these rules in any order will eventually lead to an unsimplifiable type.
%\end{lemma}

We first provide a proof of Lemma~\ref{lem:simplify}.
\begin{proof}[Proof of Lemma~\ref{lem:simplify}]
  Let $t \in \types$ and let $R$ be a homogeneous regular expression of type~$t$. For each claimed simplification rule (from $t$ to some type $t'$) we show that there is an easy transformation of $R$ into a regular expression $R'$ which is homogeneous of type $t'$ and describes the same language as~$R$ (except for rule 3, where the language is slightly changed by removing the empty string). This transformation can be performed in linear time. Together with a similar transformation in the opposite direction, this yields an equivalence of $t$-membership and $t'$-membership under linear-time reductions.
  
  (1) Suppose that $t$ contains a substring $pp$, say $t_i = t_{i+1} = p \in \{\circ,\pipe,\star,\plus\}$, and denote by $t'$ the simplified type, resulting from $t$ by deleting the entry $t_{i+1}$. In $R$, for any node on level $i$ put a direct edge to any descendant in level $i+2$ and then delete all internal nodes in level $i+1$. This yields a regular expression $R'$ of type $t'$. For any $p \in \{\circ,\pipe,\star,\plus\}$ it is easy to see that both regular expressions descibe the same language. In particular, this follows from the facts $(E^\star)^\star = E^\star$ for any regular expression $E$ (and similarly for $\plus$) and $(E_{1,1} \circ \ldots \circ E_{1,k(1)}) \circ \ldots \circ (E_{\ell,1} \circ \ldots \circ E_{\ell,k(\ell)}) = E_{1,1} \circ \ldots \circ E_{1,k(1)} \circ \ldots \circ E_{\ell,1} \circ \ldots \circ E_{\ell,k(\ell)}$ for any regular expressions $E_{i,j}$ (and similarly for $\pipe$). This yields a linear-time reduction from $t$-membership to $t'$-membership. For the opposite direction, if the $i$-th layer is labeled $p$ then we may introduce a new layer between $i-1$ and $i$ containing only degree 1 nodes, labelled by $p$. This means we replace $E^\star$ by $(E^\star)^\star$, and similarly for $\plus$. For $\circ$ and $\pipe$ the degree 1 vertices are not visible in the written form of regular expressions\footnote{\label{foot:discussion}This is the only place where we need to use degree 1 vertices; all other proofs in this paper also work with the additional requirement that each $\pipe$- or $\circ$-node has degree at least 2 (cf. footnote~\ref{foot:shortdiscussion} on page~\pageref{foot:shortdiscussion}). Note that for our construction here in Lemma~\ref{lem:simplify} it is essential to use degree 1 nodes: Two levels of $\circ$-nodes with each node having degree at least 2 have at least 4 children, so it seems to be impossible to embed a regular expression of the shorter type (where the degree could be 2 or 3) into a regular expression of the longer type (where the combined degree is at least 4). For $\pipe$-nodes, there is an easy trick by adding dummy leaves that are labeled by fresh symbols not appearing in the input string $s$, cf. footnote~\ref{foot:getridofor} on page \pageref{foot:getridofor}. For $\circ$-nodes, we are not aware of any similar trick. 
  
  An option to still make our proofs work is to consider homogeneous regular expressions of type $t$ where each $\circ$- or $\pipe$-node has degree at least $d$, giving rise to the $(t,d)$-membership problem. 
  Then $(t,d)$-membership has the same complexity as the problem $t$-membership studied in this paper, irrespective of the (constant) value of $d$. 
  To show this, for each reduction rule from $t$-membership to $t'$-membership of Lemmas~\ref{lem:simplify} and~\ref{lem:subsequencehardness} one can easily establish a similar reduction from $(t,d)$-membership to $(t',d')$-membership \emph{for some $d' \ge 2$} depending only on $t,t',d$.
  Moreover, since $(t,d)$-membership is a special case of $(t,d')$-membership for any $d \ge d'$, it suffices that we prove all algorithmic results for $(t,2)$-membership, and all negative results for $(t,d)$-membership for any constant $d\ge 2$ to show that the complexity of $(t,d)$-membership does not depend on $d$. As our algorithms even work for degree 1 nodes, the former is satisfied. For the latter, it suffices to observe that in the direct hardness proofs of Backurs and Indyk~\cite{backursindyk} and our Theorems~\ref{thm:hardnessone}--\ref{thm:hardnessthree}, all degrees grow with the input size, and thus any constant lower bound does not break the reduction. This is a viable option for proving our results in the more restricted setting with degrees at least two, but we think that it would obscure the overall point of the paper.}.

  (2) For any regular expressions $E_1,\ldots,E_k$, the expression $\big((E_1^\plus) \pipe \ldots \pipe (E_k^\plus)\big)^\plus$ describes the same language as $\big(E_1 \pipe \ldots \pipe E_k\big)^\plus$. Indeed, any string described by the former expression can be written as a concatenation of some number of strings  in the union over all languages described by $E_1,\ldots,E_k$, which is exactly the language described by the latter expression.
  Thus, the inner $\plus$\nobreakdash-operation is redundant and can be removed. Specifically, for a homogeneous regular expression of type $t$ with $t_{i},t_{i+1},t_{i+2} = \plus \pipe \plus$ we may contract all edges between layer $i+1$ and $i+2$ to obtain a homogeneous regular expression $R'$ of type $t'$ describing the same language as $R$. This yields a linear-time reduction from $t$-membership to $t'$-membership. For the opposite direction, we may introduce a redundant $\plus$-layer below any $\plus \pipe$-layers without changing the language. 
  
  (3) Note that for any regular expression $E$ we can check in linear time whether it describes the empty string, by a simple recursive algorithm: No leaf describes the empty string. Any $\star$-node describes the empty string. A $\plus$-node describes the empty string if its child does so. A $\pipe$-node describes the empty string if at least one of its children does so. A $\circ$-node describes the empty string if all of its children do so. Perform this recursive algorithm to decide whether the root describes the empty string.
  
  Now suppose that $t$ has prefix $r \star$ for some $r \in \{\plus,\pipe\}^*$, and denote by $t'$ the type where we replaced the prefix $r \star$ by $r \plus$. (Note that we could restrict our attention to the case where $r$ is a subsequence of $\pipe \plus \pipe$, since otherwise rules 1 or 2 apply.) Let $R$ be homogeneous of type $t$, and $s$ a string for which we want to decide whether it is described by $R$. If $s$ is the empty string, then as described above we can solve the instance in linear time. Otherwise, we adapt $R$ by labeling any internal node in layer $|r|+1$ by $\plus$, obtaining a homogeneous regular expression~$R'$ of type $t'$. Then $R$ describes $s$ if and only $R'$ describes $s$. Indeed, since $\star$ allows more strings than $\plus$, $R$ describes a superset of $R'$. Moreover, if $R$ describes $s$, then we trace the definitions of $\pipe$ and $\plus$ as follows. We initialize node $v$ as the root and string $x$ as $s$. If the current vertex $v$ is labelled $\pipe$, then the language of $v$ is the union over the children's languages, so the current string $x$ is contained in the language of at least one child $v'$, and we recursively trace $(v',x)$. If $v$ is labelled $\plus$, then we can write $x$ as a concatenation $x_1 \ldots x_k$ for some $k \ge 1$ and all $x_i$ in the language of the child $v'$ of $v$. Note that we may remove all $x_\ell$ which equal the empty string. For each remaining $x_\ell$ we trace $(v',x_\ell)$. Running this traceback procedure until layer $|r|+1$ yields trace calls $(v_i,x_i)$ such that $v_i$ describes $x_i$ for all~$i$, and the $x_i$ partition $s$. 
  Since by construction each $x_i$ is non-empty, $x_i$ is still in the language of $v_i$ if we relabel $v_i$ by $\plus$. This shows that $s$ is also described by $R'$, where we relabeled each node in layer $|r|+1$ by $\plus$. %(Note that the same argument would not work for $\circ$-layers, since in an expression $E_1 \circ E_2$ it may not be possible to ignore the empty string in the languages of $E_1$ and $E_2$.)
  
  Finally, applying these rules eventually leads to an unsimplifiable type, since rules 1 and 2 reduce the length of the type and rule 3 reduces the number of $\star$-operations. Thus, each rule reduces the sum of these two non-negative integers.
\end{proof}

\begin{lemma} \label{lem:subsequencehardness}
  For types $t,t'$, there is a linear-time reduction from $t'$-membership to $t$-membership if one of the following sufficient conditions holds:
  \begin{enumerate}
    \item $t'$ is a prefix of $t$,
    \item we may obtain $t$ from the sequence $t'$ by inserting a $\pipe$ at any position,
    \item we may obtain $t$ from the sequence $t'$ by replacing a $\star$ by $\plus \star$, or
    \item $t'$ starts with $\circ$ and we may obtain $t$ from the sequence $t'$ by prepending a $\plus$.
  \end{enumerate}
\end{lemma}

\begin{proof}
  (1) The definition of ``homogeneous with type $t$'' does not restrict the appearance of leafs in any level. Thus, any homogeneous regular expression of type $t'$ (i.e., the prefix) can be seen as a homogeneous regular expression of type $t$ (i.e., the longer string) where all leaves appear in levels at most $|t'|+1$. This shows that trivially $t'$-membership is a special case of $t$-membership.
  
  (2) Let $t$ be obtained from $t'$ by inserting a $\pipe$ at position $i$. %, i.e., $t_i = \pipe$, $t_j = t'_j$ for $1 \le j < i$, and $t_j = t'_{j-1}$ for $i< j \le |t'|+1$. 
  Consider any homogeneous regular expression $R'$ of type $t'$. Viewed as a tree, in $R$ we subdivide any edge from a node in layer $i-1$ to a node in layer $i$, and mark the newly created nodes by $\pipe$. This yields a regular expression $R$ of type $t$. Since a $\pipe$ with degree 1 is trivial, the language described by $R$ is the same as for $R'$.\footnote{\label{foot:getridofor}An alternative to using $\pipe$-nodes of degree 1 is as follows: Let $x$ be a fresh symbol, not occuring in the input string $s'$ for which we want to know whether it matches $R'$. For each newly created node $v$ by the subdivision process, add a new leaf node $\ell_v$ marked by $x$ and connect it to $v$. This again yields a regular expression of type~$t$. Since $x$ is a fresh symbol, the newly added leaves cannot match any symbol in $s'$, and thus $s'$ matches $R'$ if and only if $s'$ matches $R$.}
  
  (3) Let $R'$ be a homogeneous regular expression of type $t'$ with $t'_i = \star$. Subdivide any edge in $R'$ from a node in layer $i-1$ to an internal node in layer $i$, and label the newly created nodes by $\plus$. This yields a regular expression $R$ of type $t$. (Note that by this construction any leaf of $R'$ in layer $i$ stays a leaf of $R$ in layer $i$.) Consider any newly created node $v$ with child $u$. Since $u$ is an internal node, it is labeled by $\star$ and (its subtree) represents a regular expression $E^\star$. The newly created node $v$ thus represents the regular expression $(E^\star)^\plus$. Using the fact $(E^\star)^\plus = E^\star$ for any regular expression $E$, it follows that $R$ and $R'$ describe the same language. Hence, we obtain a linear-time reduction from $t'$-membership to $t$-membership.

  (4) Let $R' = E_1 \circ \ldots \circ E_k$ be a homogeneous regular expression of type $t'$. Let $s'$ be the input string for which we want to know whether it matches $R'$, and let $x$ be a fresh symbol not occuring in $s'$. We construct the expression $R := (x \circ E_1 \circ \ldots \circ E_k \circ x)^+$ and the string $s := x s' x$. Note that $R$ is homogeneous of type $t$. Moreover, there are exactly two occurences of $x$ in $s$, and thus $s$ matches $R$ if and only if $s'$ matches $E_1 \circ \ldots \circ E_k = R'$. Thus, we obtain a linear-time reduction from $t'$-membership to $t$-membership.
\end{proof}

We are now ready to prove the dichotomy theorem for the membership problem.

\begin{proof}[Proof of Theorem~\ref{thm:main}]
  The natural way of enumerating all types yields a tree with vertex set $\types$, where types $t$ and $t'$ are connected by an edge if we obtain $t$ from $t'$ by appending one of the operations $\circ,\pipe,\star,\plus$. To keep this tree simple, we directly apply the simplification rule Lemma~\ref{lem:simplify} (1), i.e., we do not consider types with consecutive equal operations. We split this tree into the three Figures~\ref{fig:membcirc}-\ref{fig:membpipe} according to the first operation $\circ/\star, \plus, \pipe$. 
  
  In the following we describe these figures. If $t$-membership is solvable in almost-linear time, then in the figures we mark the node corresponding to $t$ by the fastest known running time, and we refer to \cite{backursindyk} or our Theorem~\ref{thm:newalgos}. 
  We remark that most $O(n+m)$ algorithms are immediate, but for completeness we refer to~\cite{backursindyk}.
  
  If $t$-membership simplifies, then $t$-membership is equivalent to $t'$-membership for some different, non-simplifying type $t'$ in the tree.
  In this case, we mark $t$ by the corresponding simplification rule of Lemma~\ref{lem:simplify}. Note that the simplification rules have the property that if $t$ simplifies then any descendant of $t$ in the tree also simplifies. Thus, we may ignore the subtree of any simplifying type $t$.
  
  If we can show that any algorithm for $t$-membership takes time $(nm)^{1-o(1)}$ unless SETH fails, then in the figures we mark $t$ as ``hard''. Note that by Lemma~\ref{lem:subsequencehardness} (1), if $t$ is hard then any descendant of $t$ in the tree is also hard. Thus, we may ignore the subtree of $t$, and we only mark minimal hardness results. 
  From the results by Backurs and Indyk~\cite{backursindyk} (Theorem~\ref{thm:backursindyklower}) we know that $t$-membership takes time $(nm)^{1-o(1)}$ under SETH for the types $\circ \star$, $\circ \pipe \circ$, $\circ \plus \circ$, $\circ \pipe \plus$, and $\circ \plus \pipe$. Our Theorems~\ref{thm:hardnessone}-\ref{thm:hardnessthree} add the types $\plus \pipe \circ \plus$, $\plus \pipe \circ \pipe$, and $\pipe \plus \pipe \circ$.
  If there is such a direct hardness proof of $t$-membership, then in the figures we refer to the corresponding theorem. In all other minimal hard cases, there is a combination of the reduction rules, Lemma~\ref{lem:subsequencehardness} (2--4), resulting in a type $t'$ such that hardness of $t'$-membership follows from Theorem~\ref{thm:backursindyklower} and $t'$-membership has a linear time reduction to $t$-membership. In this case, in the figures we additionally mark the node corresponding to $t$ by $t'$. E.g., $\pipe \circ \pipe \star$-membership contains $\circ \star$-membership as a special case (since by Lemma~\ref{lem:subsequencehardness} (2) we may remove any $\pipe$ operations) and $\circ \star$-membership is hard by Theorem~\ref{thm:backursindyklower}, so we mark the node corresponding to $\pipe \circ \pipe \star$ by ``hard $\circ \star$''. 
  
  It is easy to check that our Figures~\ref{fig:membcirc}-\ref{fig:membpipe} indeed enumerate all cases and thus contain all maximal algorithmic results and minimal hardness results. The claimed dichotomy of Theorem~\ref{thm:main} now follows by inspecting these figures.
\end{proof}

\begin{figure}[hp]
   \includegraphics[scale=0.8]{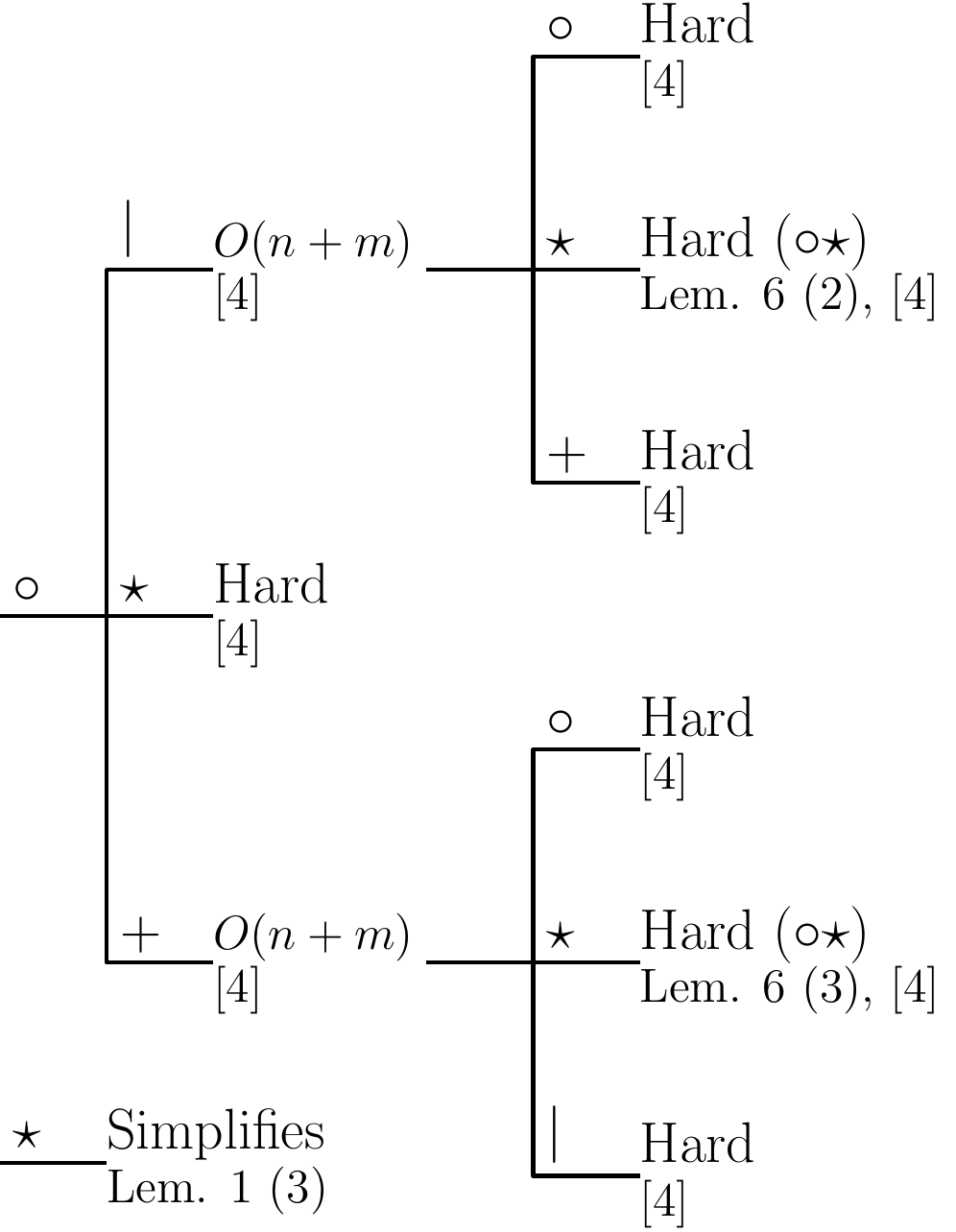}
   \caption{Tree diagram for $t$-membership with first operation $\circ$ or $\star$. For details see the proof of Theorem~\ref{thm:main}.}
\label{fig:membcirc}
\end{figure}

\begin{figure}[hp]
  \includegraphics[scale=0.8]{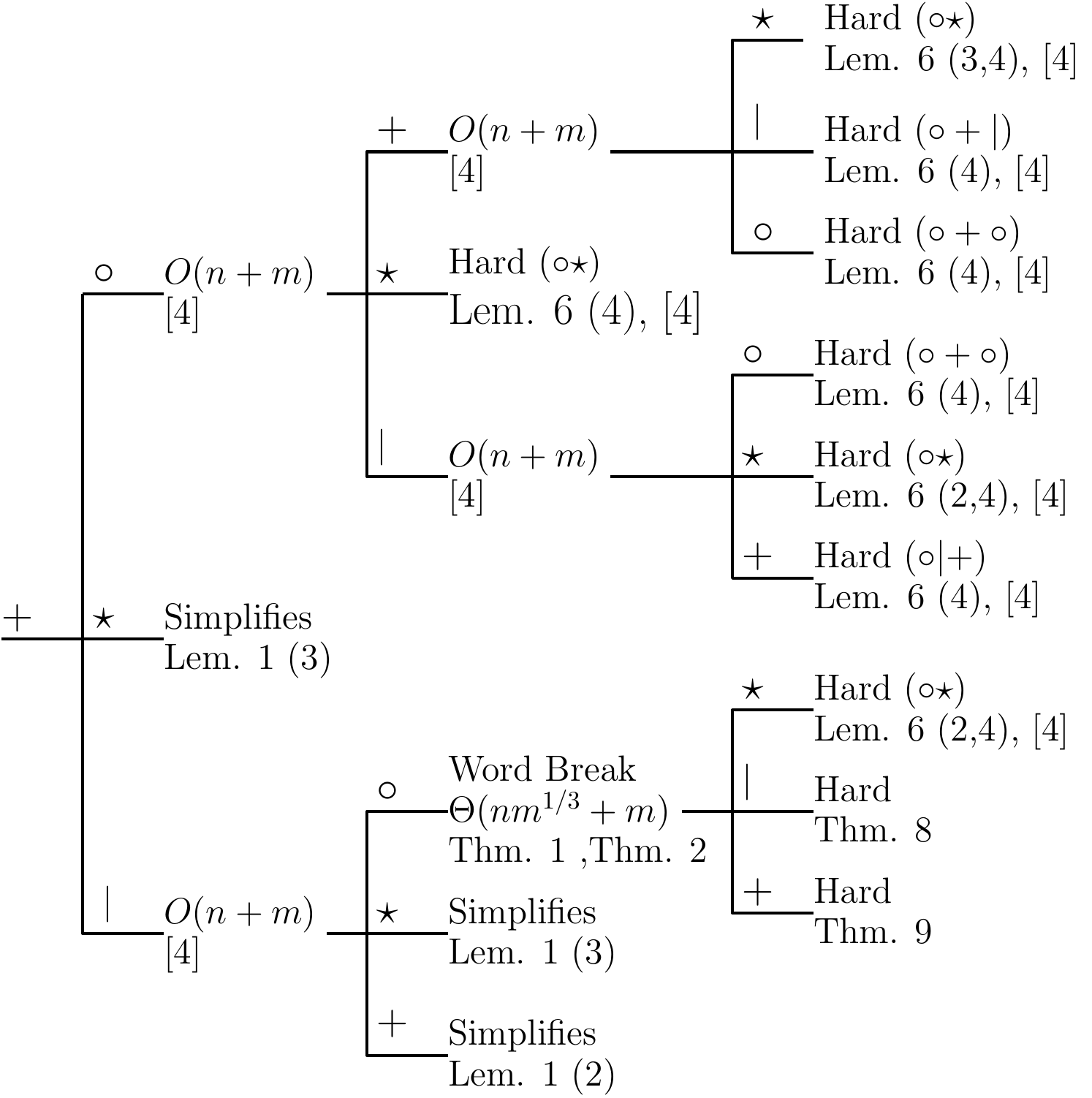}  
  \caption{Tree diagram for $t$-membership with first operation $\plus$. For details see the proof of Theorem~\ref{thm:main}.}
  \label{fig:membplus}
\end{figure}

\begin{figure}[hp]
  \includegraphics[scale=0.8]{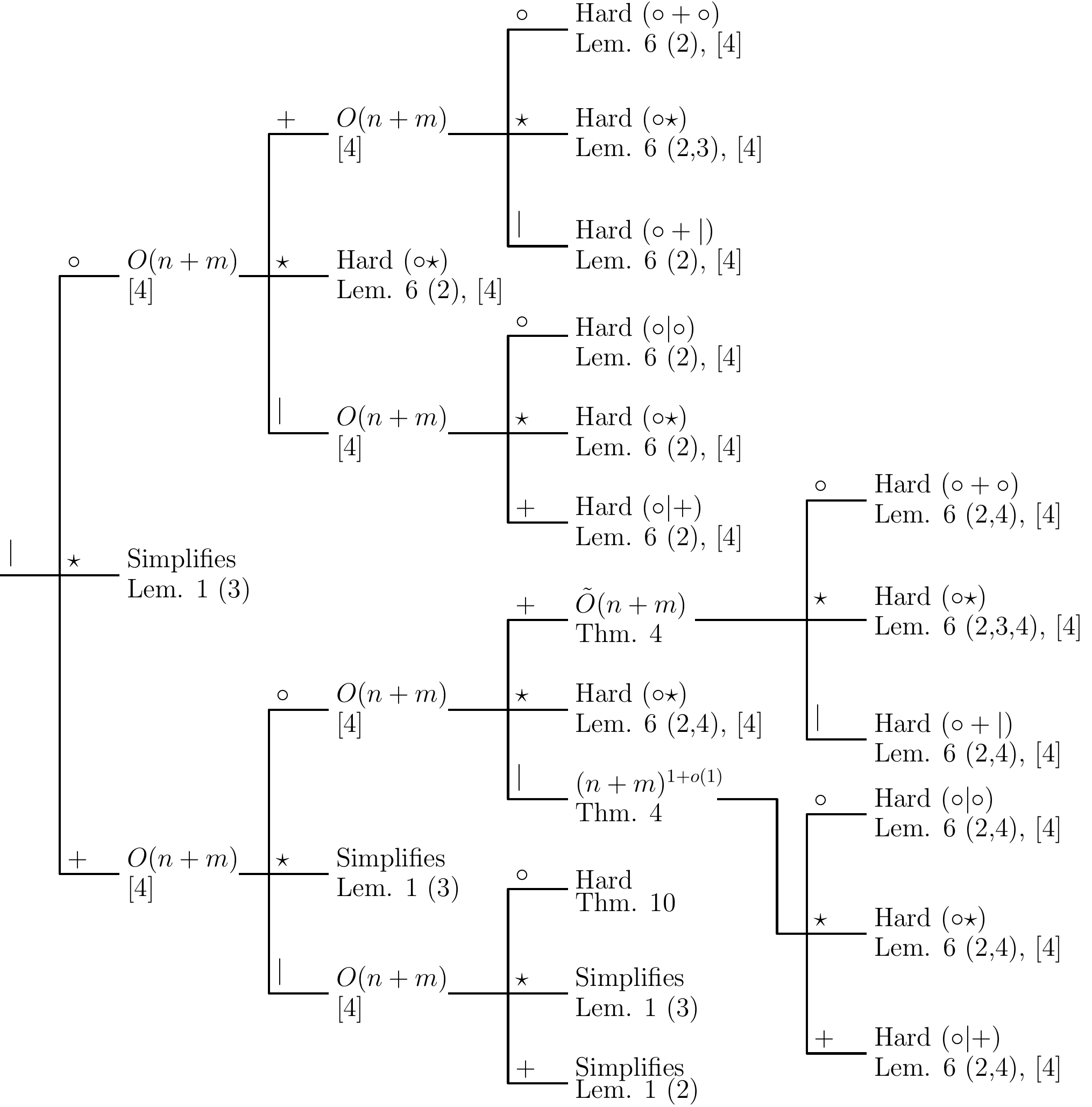}
  \caption{Tree diagram for $t$-membership with first operation $\pipe$. For details see the proof of Theorem~\ref{thm:main}.}
  \label{fig:membpipe}
\end{figure}

\clearpage
\bibliographystyle{plain}
\bibliography{pattern.bib}

\begin{thebibliography}{10}

\bibitem{AW14}
Amir Abboud and Virginia~Vassilevska Williams.
\newblock Popular conjectures imply strong lower bounds for dynamic problems.
\newblock In {\em Proceedings of the 2014 IEEE 55th Annual Symposium on
  Foundations of Computer Science}, pages 434--443, 2014.

\bibitem{AC75}
Alfred~V. Aho and Margaret~J. Corasick.
\newblock Efficient string matching: An aid to bibliographic search.
\newblock {\em Commun. ACM}, 18(6):333--340, June 1975.

\bibitem{abboudLowerBoundForValiantsParser}
Virginia Vassilevska~Williams Amir~Abboud, Arturs~Backurs.
\newblock If the current clique algorithms are optimal, so is valiant's parser.
\newblock In {\em 56th Annual IEEE Symposium on Foundations of Computer
  Science}, 2015.

\bibitem{backursindyk}
Arturs Backurs and Piotr Indyk.
\newblock Which regular expression patterns are hard to match?
\newblock In {\em 57th Annual IEEE Symposium on Foundations of Computer
  Science}, 2016.

\bibitem{BT09}
Philip Bille and Mikkel Thorup.
\newblock Faster regular expression matching.
\newblock In {\em 36th International Colloquium on Automata, Languages and
  Programming}, pages 171--182, 2009.

\bibitem{CIP09}
Raphael Clifford.
\newblock Matrix multiplication and pattern matching under hamming norm.
\newblock
  http://www.cs.bris.ac.uk/Research/Algorithms/events/BAD09/BAD09/Talks/BAD09-Hammingnotes.pdf.

\bibitem{CH02}
Richard Cole and Ramesh Hariharan.
\newblock Verifying candidate matches in sparse and wildcard matching.
\newblock In {\em Proceedings of the Thiry-fourth Annual ACM Symposium on
  Theory of Computing}, pages 592--601, 2002.

\bibitem{EG04}
Friedrich Eisenbrand and Fabrizio Grandoni.
\newblock On the complexity of fixed parameter clique and dominating set.
\newblock {\em Theoretical Computer Science}, 326(1):57--67, 2004.

\bibitem{AnkaOvermars}
Anka Gajentaan and Mark~H. Overmars.
\newblock On a class of o(n2) problems in computational geometry.
\newblock {\em Comput. Geom. Theory Appl.}, 5(3):165--185, October 1995.

\bibitem{GL12}
Fran{\c{c}}ois~Le Gall.
\newblock Faster algorithms for rectangular matrix multiplication.
\newblock In {\em 53rd Annual {IEEE} Symposium on Foundations of Computer
  Science}, pages 514--523, 2012.

\bibitem{Gronwall}
Thomas~Hakon Gr\"onwall.
\newblock Some asymptotic expressions in the theory of numbers.
\newblock {\em Trans. Amer. Math. Soc.}, 14:113--122, 1913.

\bibitem{ImpagliazzoPaturi}
Russell Impagliazzo and Ramamohan Paturi.
\newblock On the complexity of k-sat.
\newblock {\em J. Comput. Syst. Sci.}, 62(2):367--375, March 2001.

\bibitem{ImpagliazzoPZ01}
Russell Impagliazzo, Ramamohan Paturi, and Francis Zane.
\newblock Which problems have strongly exponential complexity?
\newblock {\em J. Computer and System Sciences}, 63(4):512--530, 2001.

\bibitem{KMP77}
D.~E. Knuth, J.~H. Morris, and V.~B. Pratt.
\newblock Fast pattern matching in strings.
\newblock {\em SIAM Journal on Computing}, 6:323--350, 1977.

\bibitem{Myers92}
Gene Myers.
\newblock A four russians algorithm for regular expression pattern matching.
\newblock {\em J. ACM}, 39(2):432--448, April 1992.

\bibitem{PatWil10}
Mihai P\u{a}tra\c{s}cu and Ryan Williams.
\newblock On the possibility of faster sat algorithms.
\newblock In {\em Proceedings of the Twenty-first Annual ACM-SIAM Symposium on
  Discrete Algorithms}, pages 1065--1075, 2010.

\bibitem{Thompson68}
Ken Thompson.
\newblock Programming techniques: Regular expression search algorithm.
\newblock {\em Commun. ACM}, 11(6):419--422, June 1968.

\bibitem{vw09}
Virginia Vassilevska.
\newblock Efficient algorithms for clique problems.
\newblock {\em Inf. Process. Lett.}, 109(4):254--257, January 2009.

\bibitem{Weiner73}
Peter Weiner.
\newblock Linear pattern matching algorithms.
\newblock In {\em Proceedings of the 14th Annual Symposium on Switching and
  Automata Theory}, pages 1--11, 1973.

\bibitem{WI1906}
S.~Wigert.
\newblock Sur l'ordre de grandeur du nombre des diviseurs d'un entier.
\newblock {\em Ark. Math. Astron. Fys.}, 3:113--140, 1906-07.

\bibitem{RW05}
Ryan Williams.
\newblock A new algorithm for optimal 2-constraint satisfaction and its
  implications.
\newblock {\em Theor. Comput. Sci.}, 348(2):357--365, December 2005.

\bibitem{WAPSP}
Virginia~Vassilevska Williams and Ryan Williams.
\newblock Subcubic equivalences between path, matrix and triangle problems.
\newblock In {\em Proceedings of the 2010 IEEE 51st Annual Symposium on
  Foundations of Computer Science}, pages 645--654, 2010.

\end{thebibliography}

\appendix
 
% !TEX root = main.tex

 \section{Dichotomy for Pattern Matching}
 \label{app:patternmatching}
 
For the regular expression pattern matching problem Backurs and Indyk~\cite{backursindyk} characterized all types of length at most 3. In this section, we show that their results in fact imply a complete dichotomy for all types of any (constant) depth.

Recall that in the pattern matching problem we are given a regular expression $R$ and a string~$s$ and want to decide whether \emph{any substring of $s$} is in the language described by $R$ (whereas in the membership problem we want to know whether the whole string $s$ is in the language of~$R$). 
For $t$-pattern matching, i.e., the natural restriction of pattern matching to regular expressions~$R$ that are homogeneous of type $t$, Backurs and Indyk established that any algorithm takes time $(nm)^{1-o(1)}$ (unless SETH fails) for any type $t$ among $\circ \star$, $\circ \pipe \circ$, $\circ \plus \circ$, $\circ \pipe \plus$, $\circ \plus \pipe$, $\pipe \circ \pipe$, and $\pipe \circ \plus$. Positive results for $t$-pattern matching date back to the 70ies, since $\circ$-pattern matching corresponds to standard string matching, which has a classic linear-time algorithm due to Knuth, Morris, and Pratt~\cite{KMP77}. Further special cases with near-linear time algorithms are $\pipe \circ$ (dictionary matching~\cite{AC75}), $\circ \pipe$ (superset matching~\cite{CH02}), and $\circ \plus$ (\cite{backursindyk}). 
An easy observation moreover shows that $t$-pattern matching is in linear time whenever $t$ starts with $\star$ or $\pipe \star$.

\begin{lemma} \label{lem:PMtrivial}
  For any type $t$ with prefix $\star$ or with prefix $\pipe \star$, $t$-pattern matching is in linear time.
\end{lemma}
\begin{proof}
  In both cases, any homogeneous regular expression $R$ of type $t$ matches the empty string, and thus the empty substring of any string $s$ is in the language described by $R$, so that any string $s$ is a YES-instance -- unless all leaves of $R$ are too high, so that there is no internal $\star$-node. In the latter case, either the root of $R$ is a leaf, so that the language described by $R$ consists of a single string of length one, or the root is a $\pipe$-node all of whose children are leaves, so that the language described by $R$ is a finite set of strings of length one. In both cases, it is easy to check in linear time whether (any substring of) a given string $s$ matches $R$.
\end{proof}

For pattern matching, the following two lemmas show similar simplification rules and reduction rules as for the membership problem (cf.~Lemmas~\ref{lem:simplify} and \ref{lem:subsequencehardness}).

\begin{lemma} \label{lem:PMsimplify}
  For any type $t$, applying any of the following rules yields a type $t'$ such that $t$-pattern matching and $t'$-pattern matching are equivalent under linear-time reductions:
  \begin{enumerate}
    \item replace any substring $p p$, for some $p \in \{\circ,\pipe,\star,\plus\}$, by $p$,
    \item remove prefix $\plus$,
    \item replace prefix $\pipe \plus$ by $\pipe$.
  \end{enumerate}
  We say that $t$-pattern matching \emph{simplifies} if one of these rules applies.
  Applying these rules in any order will eventually lead to an unsimplifiable type.
\end{lemma}
\begin{proof}
  Rule 1 also holds for the membership problem. The same proof as in Lemma~\ref{lem:simplify} works verbatim for pattern matching, since we argued that we can turn any homogeneous regular expression $R$ of type $t$ into one of type $t'$ without changing its described language, and without looking at the input string $s$ (and similarly the other way round).
  
  For rule 2, note that string $s$ contains a substring matching regular expression $E^\plus$ if and only if it contains a substring matching $E$. Indeed, the former means that $s$ has a substring $s' = s'_1 \ldots s'_k$ with each $s'_i$ matching $E$, but then any $s'_i$ is a substring of $s$ matching $E$. The opposite direction holds since $E^+$ describes a superset of the language of $E$. Thus, we may remove any prefix $\plus$, and the inverse operation of introducing a redundant prefix $\plus$ also does not change the problem.
  
  For rule 3, we similarly have that string $s$ contains a substring matching regular expression $E_1^\plus \pipe \ldots \pipe E_k^\plus$ if and only if it contains a substring matching $E_1 \pipe \ldots \pipe E_k$.
\end{proof}

\begin{lemma} \label{lem:PMsubsequencehardness}
  For types $t,t'$, there is a linear-time reduction from $t'$-membership to $t$-membership if one of the following sufficient conditions holds:
  \begin{enumerate}
    \item $t'$ is a prefix of $t$,
    \item we may obtain $t$ from the sequence $t'$ by replacing a $\star$ by $\plus \star$, or
    \item we may obtain $t$ from the sequence $t'$ by inserting a $\pipe$ at any position.
    %\item $t'$ starts with $\circ$ and we may obtain $t$ from the sequence $t'$ by prepending a $\plus$.
  \end{enumerate}
\end{lemma}
\begin{proof}
  The proof of Lemma~\ref{lem:subsequencehardness} works verbatim.
\end{proof}

We are now ready to prove the following dichotomy.

\begin{theorem} \label{thm:patternmatching}
  For any $t \in \types$ one of the following holds:
  \begin{itemize}
    \item $t$-pattern matching simplifies,
    \item $t$ has prefix $\star$ or $\pipe \star$ and thus $t$-pattern matching is in linear time (Lemma~\ref{lem:PMtrivial}),
    \item $t$ is a subsequence of $\pipe \circ$ (dictionary matching~\cite{AC75}),  $\circ \pipe$ (superset matching~\cite{CH02}), or $\circ \plus$ (\cite{backursindyk}), and thus $t$-membership is in near-linear time, or
    \item $t$-membership takes time $(nm)^{1-o(1)}$, assuming SETH.
  \end{itemize}
\end{theorem}
\begin{proof}
  All algorithmic results and minimal hardness results were known before, we just show that the known results are sufficient to completely characterize all types of any (constant) depth. The proof is similar to the one of Theorem~\ref{thm:main}. Again we consider a tree containing all types not containing two consecutive equal operations (i.e., not simplifiable by Lemma~\ref{lem:PMsimplify}.1), see Figure~\ref{fig:patternmatching}. 
  
  When one of the simplification rules of Lemma~\ref{lem:PMsimplify} applies to $t$, then $t$-pattern matching is equivalent to $t'$-pattern matching for some different, unsimplifiable $t'$ in the tree. Since the same simplification rule also applies to all descendants of $t$ in the tree, we can ignore the whole subtree of~$t$.
  
  For the subtrees starting with $\star$ or $\pipe \star$ we know that $t$-pattern matching is in linear time for each type $t$ by Lemma~\ref{lem:PMtrivial}. 
  
  For any other type $t$ with a near-linear time algorithm, in the figure we annotate the corresponding node by the fastest known asymptotic running time. 
  
  Finally, when there is a SETH-based lower bound for $t$-pattern matching then the same lower bound also holds for all descendants of $t$ in the tree (by Lemma~\ref{lem:PMsubsequencehardness}.1), so we have characterized the whole subtree of $t$. When hardness directly follows from one of the reductions in~\cite{backursindyk} then we simply mark $t$ as ``hard''. When we first have to use one of the reductions in Lemma~\ref{lem:PMsubsequencehardness} to reduce from a hard type $t'$ of~\cite{backursindyk}, then in the figure we annotate $t$ by ``hard $t'$'' as well as the corresponding reduction rule.
  
  It is easy to check that our Figure~\ref{fig:patternmatching} indeed enumerates all cases and thus contains all maximal algorithmic results and minimal hardness results. The claimed dichotomy of Theorem~\ref{thm:patternmatching} now follows by inspecting the figure.
\end{proof}

%% \begin{figure}[hp]
%% \includegraphics[scale=0.8]{patternmatching.pdf}
%% \caption{Tree diagram for $t$-pattern matching. For details see the proof of Theorem~\ref{thm:patternmatching}.}
%% \label{fig:patternmatching}
%% \end{figure}
\begin{figure}[hp]
\includegraphics[scale=0.8]{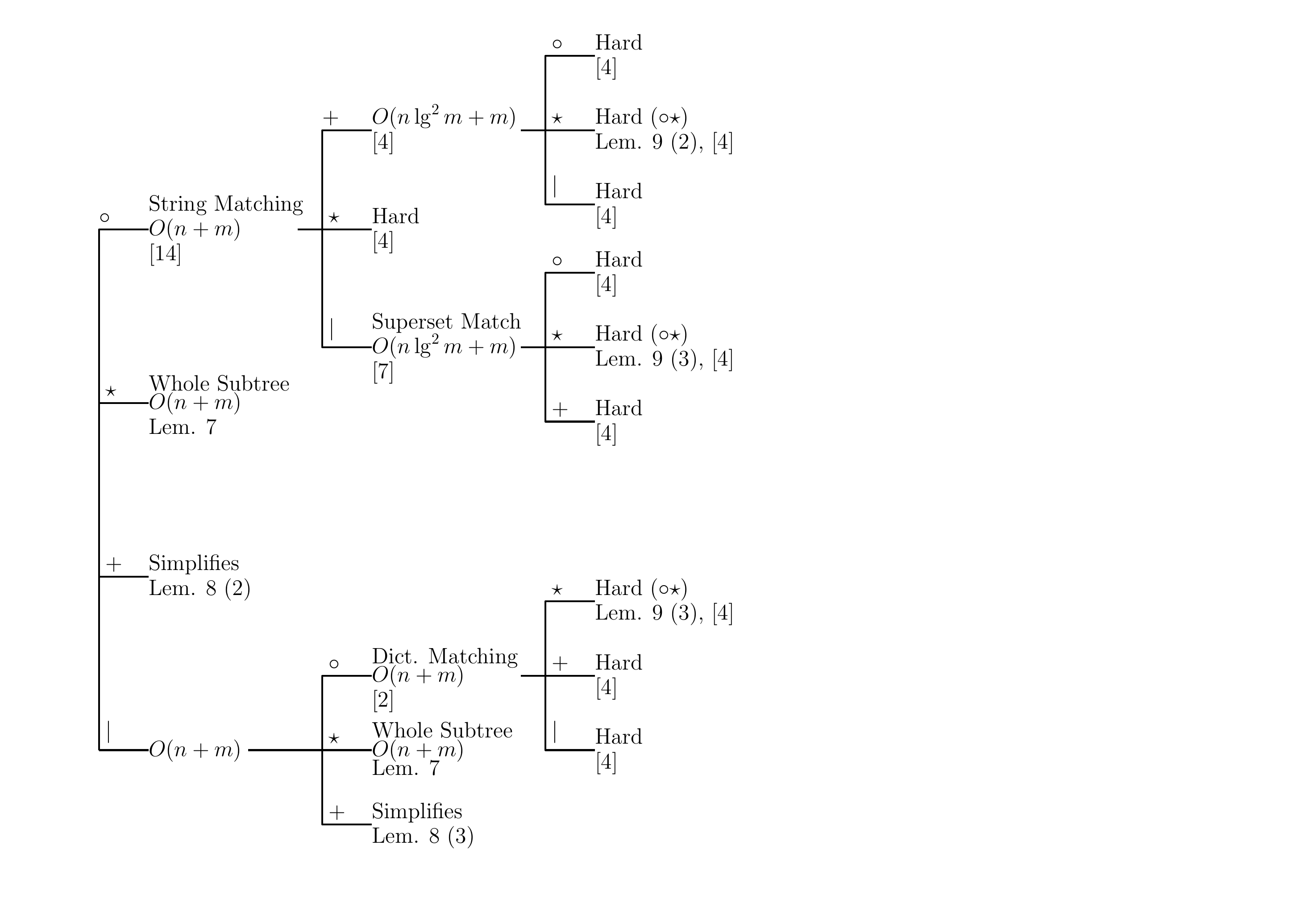}
\caption{Tree diagram for $t$-pattern matching. For details see the proof of Theorem~\ref{thm:patternmatching}.}
\label{fig:patternmatching}
\end{figure}

\end{document}